%% file: paper.tex
\documentclass[runningheads]{llncs}

\newif\ifaplas
\newcommand{\aplasorlong}[2]{\ifaplas{#1}\else{#2}\fi}

\usepackage{graphicx}
\usepackage{wrapfig}
\usepackage{scrextend}
\usepackage[iso]{datetime}
\usepackage[utf8]{inputenc}
\usepackage[british]{babel}
\usepackage{float}
\usepackage{newfloat}
\usepackage{framed}
\usepackage{amssymb}
\usepackage{amsmath}
\usepackage{mathtools}
\usepackage{braket}
\usepackage{color}
\usepackage{subcaption,chapterbib}
\usepackage[numbers,comma,sectionbib]{natbib}
\usepackage{authblk}
\usepackage{algorithm}
\usepackage[noend]{algpseudocode}
\usepackage{resizegather}

\usepackage{stmaryrd}
\SetSymbolFont{stmry}{bold}{U}{stmry}{m}{n}

\usepackage{bussproofs}
\usepackage{listingsutf8}
\usepackage[bookmarks]{hyperref}
\usepackage[normalem]{ulem}
\usepackage{proof}
\usepackage{xfrac}
\usepackage{pdfcomment}
\usepackage{svg}
\svgpath{{figures/}}

\usepackage{tikzit}
\input{tikzstyles.tex}

\input{tikzfigures.tex}
\input{zx.tikzdefs}
\input{zx.tikzstyles}

\input{defs.tex}
\begin{document}
\title{Hybrid quantum-classical circuit simplification with the ZX-calculus}
%
%\titlerunning{Abbreviated paper title} If the paper title is too long for the
%running head, you can set an abbreviated paper title here
%
\author{Agustín Borgna\inst{1,2}\orcidID{0000-0002-1688-1370} \and Simon
Perdrix\inst{1}\orcidID{0000-0002-1808-2409} \and Benoît
Valiron\inst{3}\orcidID{0000-0002-1008-5605}}
\authorrunning{A. Borgna et al.} % First names are abbreviated in the running
%head.  If there are more than two authors, 'et al.' is used.
%
\institute{ CNRS LORIA, Inria-MOCQUA, Université de Lorraine, F 54000 Nancy,
France \and Université Paris-Saclay, CNRS, Laboratoire Méthodes Formelles,
91405, Orsay, France \and École CentraleSupélec, Laboratoire Méthodes
Formelles, 91405, Orsay, France }
\maketitle              % typeset the header of the contribution
\begin{abstract} We present a complete optimization procedure for hybrid
  quantum-classical circuits with classical parity logic.
  While common optimization techniques for quantum algorithms focus on rewriting
  solely the pure quantum segments, there is interest in applying a global
  optimization process for applications such as quantum error correction and
  quantum assertions.
  This work, based on the pure-quantum circuit optimization procedure by Duncan
  et al., uses an extension of the formal graphical ZX-calculus called
  \zxGNDsafe\ as an intermediary representation of the hybrid circuits to allow
  for granular optimizations below the quantum-gate level.
  We define a translation from hybrid circuits into diagrams that admit the
  graph-theoretical focused-gFlow property, needed for the final extraction back
  into a circuit.
  We then derive a number of gFlow-preserving optimization rules for
  \zxGNDsafe diagrams that reduce the size of the graph, and devise a strategy
  to find optimization opportunities by rewriting the diagram guided by a Gauss
  elimination process.
  Then, after extracting the circuit, we present a general procedure for
  detecting segments of circuit-like \zxGNDsafe\ diagrams which can be
  implemented with classical gates in the extracted circuit.
  We have implemented our optimization procedure as an extension to the
  open-source python library PyZX.

\end{abstract}

\section{Introduction}%
\label{sec:intro}

The description of quantum algorithms commonly involves quantum operations
interacting with classical data in its inputs, outputs, or intermediary steps
via measurements or state preparations.
Some applications such as quantum error correction~\cite{duncan_verifying_2014,
de_beaudrap_zx_2020} and quantum assertions~\cite{Li_assertions, Zhou_2019}
explicitly introduce classical measurements and logic between quantum
computations.
In general, quantum programming languages usually allow for measurements and
classically controlled quantum operators mixed-in with unitary gates
\citep{Green2013quipper,cross2017openQASM,khammassi2018cQASM,Steiger2018projectQ}.
Furthermore, Jozsa~\cite{jozsa_mbqc} conjectured that any polynomial-time
quantum algorithm can be simulated by polylogarithmic-depth quantum computation
interleaved with polynomial-depth classical computation. As such, there is
interest in contemplating this kind of structures in circuits.

A popular alternative representation of quantum circuit is based on the
\textit{ZX-calculus}~\citep{coecke_interacting_2008,
bob_coeke_aleks_kissinger_picturing_2017}, a formal diagrammatic language which
presents a more granular representation of quantum circuits and has been
successfully used in applications such as MBQC~\cite{duncan2010rewriting},
quantum error correction~\cite{chancellor2016coherent}, and quantum foundations.
Carette et al.~\cite{CJPV19completenessMix} introduced an extension of the
calculus called \zxGND\ which allows for the representation of operations
interacting with the classical environment by adding a discarding \textit{ground
generator} to the diagrams.

It is natural to look at the problem of optimizing algorithm implementations by
taking in consideration the environment in addition to the pure quantum
fragments.  However, most common optimization strategies focus solely on the
latter without contemplating the hybrid quantum-classical
structure~\citep{amy_polynomial-time_2014, heyfron_efficient_2018}. One of this
pure optimizations introduced by Duncan et al.~\cite{DKPW2019qcircSimpl} uses
the ZX-calculus to apply granular rewriting rules that ignore the boundaries of
each quantum gate. We will refer to it as the \textit{Clifford optimization}
algorithm. Their rewriting steps preserve a diagram property called
\textit{gFlow admittance} that is required for the final extraction of the
ZX diagrams into circuits. The ZX optimization method was latter used by
Kissinger and van de Wetering~\cite{aleks_kissinger_reducing_2020} in their
method to reduce the number of T-gates in quantum circuits.

In this work we define the natural extension of the pure Clifford optimization
algorithm by Duncan et al. to hybrid quantum-classical circuits using the
\zxGND\ calculus.

Our circuit optimization procedure forgets the difference between quantum and
classical wires during the simplification process, representing connections as
a single type of edge. This allows it to optimize the complete hybrid system as
an homogeneous diagram, and results in similar representations for operations
that can be done either quantumly or classically. 
Generally, in a physical quantum computer, the classical operations are simpler to
implement than their quantum counterparts, and quantum simulators can exploit
the knowledge of which wires carry classical data to simplify their operation.
As such, it is beneficial to extract classical gates in the resulting circuit
where possible.

The contribution of this paper are as follows.
\begin{itemize}
  \item We specify a translation of hybrid circuits into \zxGND\ diagrams in a
    special \textit{graph like} form that admits a gFlow, restricting the
    classical segments of the input to parity circuits.
  \item We introduce a number of gFlow-preserving rewriting rules that interact
    with the discarding generator to reduce the size of the diagrams, and devise
    a strategy to find optimization opportunities using the biadjacency matrix
    of the graph cut between spiders connected to ground generators and the
    other nodes in the diagram.
  \item We define a procedure to extract \zxGND\ diagrams with a gFlow back
    into hybrid quantum-classical circuits, including ancilla initialization and
    termination.
  \item We define the problem of \zxGND-classicalization as labelling segments
    of the diagrams which can be implemented classically and present an
    heuristic solution. Our method can be applied on the extracted circuits to
    maximize the number of classically implemented operations.
\end{itemize}

The paper is organized in the following manner.
In Section~\ref{sec:zx} we define the quantum circuits, present a syntactic
description of the \zxGND\ calculus and its equation, and give an intuition
behind the representation of hybrid quantum-classical circuits.
Section~\ref{sec:graph-like} then introduces the graph-like family of \zxGND\
diagrams and defines the focused gFlow property over the graphs.  We then define
the translation of quantum circuits into graph-like diagrams in
Section~\ref{sec:translation}.
In Section~\ref{sec:optimization} we introduce the optimization rules and our
strategy for finding rule matches which we use to describe the complete
optimization algorithm. Then in Section~\ref{sec:extraction} we define the
extraction algorithm and finally we present our classicalization procedure in
Section~\ref{sec:classicalization}.
In Section~\ref{sec:implementation} we discuss the results of testing our
procedure on randomly generated circuits.

\section{Hybrid quantum-classical circuits and the grounded ZX-calculus}%
\label{sec:zx}

In pure quantum operations, a single qubit quantum state is represented as a
unitary vector in the Hilbert space $\C^2$. We use Dirac notation to talk about
such vectors and denote an arbitrary state as $\ket{\phi}$.  States can be
be described as a linear combination of vectors in a basis such as the
computational basis $\{\ket{0}, \ket{1}\}$ or the diagonal basis $\{\ket{+},
\ket{-}\}$, where $\ket{\pm} = \frac{1}{\sqrt2}(\ket{0} \pm \ket{1})$.  A third,
less commonly used basis called $Y$ is formed by the vectors $\ketR =
\frac{1}{\sqrt2}(\ket{0} + i \ket{1})$ and $\ketL = \frac{1}{\sqrt2}(\ket{0} - i
\ket{1})$. Qubit spaces can be composed using a tensor product, and we denote
$\ket{\phi\psi} = \ket{\phi} \otimes \ket{\psi}$.

Hybrid quantum-classical systems include classical data, which can be
represented in a qubit space as orthonormal basis vectors (e.g. by representing
a logical 0 as the state $\ket{0}$ and a logical 1 as $\ket{1}$), but
additionally include a \textit{trace} or \textit{measurement} operation, which
probabilistically projects a qubit into a vector in an orthogonal basis. The
resulting probabilistic distribution of pure states is called a mixed state, and
is better represented by a \textit{density matrix}, a positive semi-definite
Hermitian operator of trace one in the $(\C^{2 \times 2})^{\otimes n}$ Hilbert
space, for an $n$-qubit system. Given a probabilistic distribution of pure
states $\{(p_i, \ket{\phi_i})\}$, their density matrix is constructed as $\sum_i
p_i \ketbra{\phi_i}{\phi_i}$, where $\bra{\phi} = \ket{\phi}^\dagger$.

Quantum circuit diagrams consist of horizontal lines carrying each the
information of one qubit, read from right to left, with some attached gates
applying unitary transformations over the qubit states. We use the universal set
of operations $\{\CNOT, \Rx[\alpha], \Z[\alpha], \H\}$ for pure-quantum
diagrams. When $\alpha$ is limited to multiples of $\frac{\pi}{4}$ this roughly
corresponds to the approximately universal Clifford+T group. Some rotation gates
have specific names, such as $\Z[] = \Z[\pi]$, $\Rx[] = \Rx[\pi]$, $\text{S} =
\Z[\frac\pi2]$, $\text{HSH} = \Rx[\frac\pi2]$, and $\text{T} = \Z[\frac\pi4]$.
We additionally include ancilla initialization and termination, and swaps.  The
representation of each mentioned gate is reproduced here.
\[ \input{circuit/cnot.tikz}
  \quad
  \input{circuit/hadamard.tikz}
  \quad
  \input{circuit/z.tikz}
  \quad
  \input{circuit/x.tikz}
  \quad
  \input{circuit/init.tikz}
  \quad
  \input{circuit/termination.tikz}
  \quad
  \input{circuit/swapQ.tikz}
\]
Hybrid circuits represent bit-carrying classical wires using doubled lines and
extend the set of gates with some classical operations such as
$\{\NOT,\XOR,\AND\}$, classical fan-out, bit swaps, measurement, qubit
preparation, and classically controlled versions of the $\Rx[\pi]$ and $\Z[\pi]$
gates. We depict them respectively as follows.
\[\scalebox{0.82}{%
  \input{circuit/not.tikz}
  \;
  \input{circuit/xor.tikz}
  \;
  \input{circuit/and.tikz}
  \;
  \input{circuit/clone.tikz}
  \;
  \input{circuit/swapC.tikz}
  \;
  \input{circuit/meas.tikz}
  \;
  \input{circuit/prepare.tikz}
  \;
  \input{circuit/notClassic.tikz}
  \;
  \input{circuit/zClassic.tikz}
}\]
Circuits are inductively constructed from these generators, wire identities, and
parallel and serial composition, ensuring that only wires of the same type
connect with each other.

In this work we restrict the input to circuits with classical parity logic,
choosing not to include AND gates due to the complexity of their representation
as \zxGND\ diagrams, which might result in the introduction of additional
non-Clifford gates during the extraction procedure (refer to
Section~\ref{sec:discussion} for further discussion).

The ZX-calculus is a formal graphical language which provides a fine-grained
representation of quantum operations. We present a brief introduction to its
definition, including the \zxGND\ extension to represent classical operations.
Refer to~\cite{vandewetering2020zxcalculus} for a complete description of both
calculi.

ZX diagrams representing pure-quantum linear maps are composed by
\textit{wires}, \textit{spiders}, and \textit{Hadamard boxes}. We read the
diagrams from right to left and represent inputs and outputs as open-ended
wires.
%\[\tikzfig{zxGnd/intro-example}\]
%
The Hadamard box $\input{zxGnd/elem/hadamard.tikz}$ swaps the computational and
diagonal basis, mapping $\ket{0}$ to $\ket{+}$, $\ket{1}$ to $\ket{-}$ and vice
versa.
The spiders are arbitrary-degree nodes labelled with a real phase $\alpha \in
[0,2\pi)$ that come in either green or red color, named Z- and X-spiders
respectively.
When $\alpha$ is a multiple $\pi$ or $\frac\pi2$, we call them \textit{Pauli}
or \textit{Clifford}-spiders respectively. 
We refer to the set of spiders connected to outputs and inputs of
the diagram as $\outputs$ and $\inputs$ respectively, and call their members
\textit{output-} and \textit{input-spiders}.

A degree-2 green (resp.\ red) spider corresponds to applying a
\Z[\alpha] (\Rx[\alpha]) operation over a qubit.  Phaseless spiders represents
nodes with phase $0$ and can be interpreted as copying the computational basis
vectors in the case of green spiders, or the diagonal basis vectors for red
spiders.
\[\input{zxGnd/elem/spiderZ-2-3.tikz}\]
Spiders of the same color can be fused together, adding their phases. It is
important to note that the relative position of the nodes in ZX diagrams do
not alter their interpretation, as only the topology matters.
%
%\[\tikzfig{zxGnd/topology}\]
%

The ZX-calculus comes equipped with a complete set of formal rewrite
rules~\cite{jeandel_completeness_2019}. We reproduce it here ignoring scalars.
\[\scalebox{0.82}{\input{zxGnd/rule/zxRules.tikz}}\]

The \zxGND\ calculus~\citep{CJPV19completenessMix} is an extension to the
ZX-calculus which is able to easily describe interactions with the environment.
The diagrams have a standard interpretation as completely positive linear
maps between quantum mixed states\aplasorlong{}{ (cf. Appendix~\ref{sec:semantics} for a
formal description)}.
In addition to the ZX generators and rewrite rules, the calculus introduces a
\textit{ground} generator ($\ground$) which represent the tracing operation, or
the discarding of information. When connected to a degree-3 green spider, this
can correspond to a measurement operation over the computational basis or a
qubit initialization from a bit.
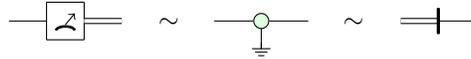
\[
  \input{circuit/meas.tikz}
  \quad\sim\quad
  \input{zxGnd/elem/spiderZgnd-1-1.tikz}
  \quad\sim\quad
  \input{circuit/prepare.tikz}
\]
We refer to the spiders attached to \ground generators as \ground-spiders.
Notice that we use the same kind of wire for both classical and quantum data,
since as previously discussed we can encode the latter as the former.
We will later introduce a method to differentiate between the two types of wire
by using the \ground-spiders in Section~\ref{sec:classicalization}.

\zxGND\ extends the set of rewriting rules with the following additions.
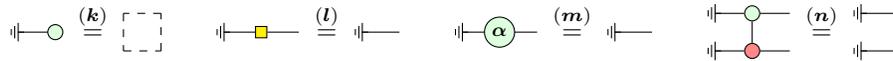
\[
  \input{zxGnd/rule/numberRule-a.tikz}
  \;\overset{\mbox{\scriptsize\GndRemoveRule}}{=}\;
  \input{unit-diagram.tikz}
  \qquad
  \input{zxGnd/rule/HRule-a.tikz}
  \;\overset{\mbox{\scriptsize\GndHadamardRule}}{=}\;
  \input{zxGnd/rule/HRule-b.tikz}
  \qquad
  \input{zxGnd/rule/MergeRule-a.tikz}
  \;\overset{\mbox{\scriptsize\GndDiscardRule}}{=}\;
  \input{zxGnd/rule/MergeRule-b.tikz}
  \qquad
  \input{zxGnd/rule/CNOTRule-a.tikz}
  \;\overset{\mbox{\scriptsize\GndCNOTRule}}{=}\;
  \input{zxGnd/rule/CNOTRule-b.tikz}
\]
Intuitively, the \ground\ generator discards any operation applied over a single
qubit.
Multiple discards can be combined into one vio the following rule, derived from
rules~\GndDiscardRule, \GndCNOTRule, and \GndRemoveRule.
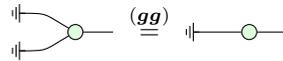
\[
  \input{zxGnd/rule/DoubleRule-a.tikz}
  \;\overset{\mbox{\scriptsize\GndDoubleRule}}{=}\;
  \input{zxGnd/rule/DoubleRule-b.tikz}
\]

For simplicity in our diagrams, we replace solely as notation the Hadamard
boxes with ``Hadamard wires'' drawn in blue, as follows.
\[ \input{zxGnd/elem/hWire.tikz} \;=\; \input{zxGnd/elem/hWireDef.tikz} \]

We introduce two additional derived equations. One to erase duplicated Hadamard
wires, as proven by Duncan et al.~\cite{DKPW2019qcircSimpl}, and another to
discard them, from a combination of rules \GndDiscardRule\ and \GndHadamardRule.
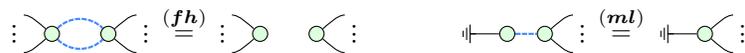
\[
  \input{zxGnd/rule/DoubleH-a.tikz}
  \;\overset{\mbox{\scriptsize\ParHRule}}{=}\;
  \input{zxGnd/rule/DoubleH-b.tikz}
  \qquad\qquad
  \input{simpl/gnd-rewriting/gndDiscard-a.tikz}
  \;\overset{\mbox{\scriptsize\GndHDiscardRule}}{=}\;
  \input{simpl/gnd-rewriting/gndDiscard-b.tikz}
\]

We utilise\ !-box notation~\citep{miller-bakewell_finite_2020} to represent
infinite families of diagrams with segments that can be repeated 0 or more times.
In the following sections it will be useful to use this notation for depicting
more complex diagrams. Here we present an example of its usage.
\[
  \scalebox{0.85}{\input{zxGnd/bangBox.tikz}}
  \;\in\;\{
  \scalebox{0.85}{\input{zxGnd/bangBox-0.tikz}}
  ,\;
  \scalebox{0.85}{\input{zxGnd/bangBox-1.tikz}}
  ,\;
  \scalebox{0.85}{\input{zxGnd/bangBox-2.tikz}}
  ,\;
  \dots
  \}
\]

\section{Graph-like diagrams and focused gFlow}%
\label{sec:graph-like}

A ZX diagram is said to be in \textit{graph-like}
form~\cite{DKPW2019qcircSimpl} when it contains only
Z-spiders connected by Hadamard wires, there are no parallel edges nor
self-loops, and no spider is connected to more than one input or output.
We define the graph-like form for \zxGND\ diagrams and include a weaker version
allowing a node to connect to an input, a ground, and any number of outputs
simultaneously. When defining a translation from quantum circuits into \zxGND\
diagrams it will be simpler to initially generate weakly graph-like diagrams and
transform the final result into the strict version afterwards.
\begin{definition}%
  \label{defn:graph-like} A \zxGND diagram is \textit{graph-like} (respectively
  \textit{weakly graph-like}) when:
  \begin{enumerate}
    \item All spiders are Z-spiders.
    \item Z-spiders are only connected via Hadamard edges.
    \item There are no parallel Hadamard edges or self-loops.
    \item There is no pair of connected \ground-spiders.
    \item Every input, output, or \ground is connected to a Z-spider.
    \item Every Z-spider connected to a \ground has phase 0.
    \item Every Z-spider is connected to
      at most one input, one output, or one \ground
      (at most one input and at most one \ground).
  \end{enumerate}
\end{definition}

\begin{proposition}%
  \label{lem:eq-to-weak-graphLike}
  Every \zxGND diagram is equivalent to a weakly graph-like \zxGND diagram.  Indeed,
  Duncan et al.~\cite{DKPW2019qcircSimpl} proved that any pure-ZX diagram is
  equivalent to a graph-like one. The proof can be extended to weakly graph-like
  \zxGND\ diagrams simply by applying rule $\GndHadamardRule$ to eliminate
  Hadamards connected to \ground generators, rule $\GndDoubleRule$ to eliminate
  duplicated \ground connected to the same spider, and rule $\GndCNOTRule$ to
  disconnect wires between \ground-spiders.
\end{proposition}

\begin{lemma}%
  \label{lem:weak-to-strict-graphlike}
  There exists an algorithm to transform an arbitrary \zxGND diagram into an
  equivalent strictly graph-like diagram.
\end{lemma}

\begin{proof}
  By adding identity spiders to the inputs and outputs. Cf.
  \aplasorlong{long version on arXiv}{Appendix~\ref{sec:appendix-proofs}}.
\end{proof}

Once a diagram is in a weakly graph-like form, all its spiders as well as all
its internal connections are of the same kind. We can refer to its underlying
structure as a simple undirected graph, marking the nodes connected to inputs
and outputs. In addition, \ground generators or the \ground-spiders connected to
them can be seen as outputs discarding information into the environment.  This
is known as the underlying open graph of a diagram.

\begin{definition}%
  \label{defn:open-graph}
  An \textit{open graph} is a triple $(G,S,T)$ where $G = (V,E)$ is an
  undirected graph, and $S \subseteq V$ is a set of \textit{sources} and $T
  \subseteq V$ is a set of \textit{sinks}.  For a weakly graph-like
  \zxGND diagram $D$, the \textit{underlying open graph} $G(D)$ is the open
  graph whose vertices are spiders $D$, whose edges correspond to Hadamard
  edges, whose set $S$ is the subset of spiders connected to the inputs of $D$,
  and whose set $T$ is the subset of spiders connected to the outputs of $D$ or
  to ground generators.
\end{definition}

The underlying open graph of a ZX diagram produced from our translation of
quantum circuits verify a graph-theoretic invariant called
\textit{focused gFlow}~\citep{mhallaWhichGraphStates2014}. This
structure ---originally conceived for graph states in measurement based quantum
computation--- gives a notion of flow of information and time on the diagram.
It will be required to guide the extraction strategy in
Section~\ref{sec:extraction}.

\begin{definition}
  Given an open graph $G$, a \textit{focused gFlow} $(g,\prec)$ on $G$ consists of a function
  $g : \overline{T} \to 2^{\overline S}$ and a partial order $\prec$ on the vertices $V$ of $G$
  such that for all $u \in \overline{T}$, $\odd_G(g(u)) \cap \overline{T} = \{u\}$
  and $\forall v \in g(u), u \prec v$
  where $2^{\overline S}$ is the powerset of $\overline{S}$ and
  $\odd_G(A) \coloneqq \{v \in V(G) \;|\; |N(v) \cap A| \equiv 1 \text{ mod } 2\}$
  is the \textit{odd neighbourhood} of $A$.
\end{definition}

\section{Translation of hybrid quantum-classical circuits}%
\label{sec:translation}

We describe our translation from hybrid quantum-classical circuits into strictly
graph-like \zxGND\ diagrams by steps. First, we translate each individual gate
directly into a weakly graph-like diagram and connect them with regular wires.
We define this translation $\toZxGND{\cdot}$ by inductively translating the
gates as described in Table~\ref{table:circ-zx-translation} immediately followed
by the application of the spider fusion rule $\SpiderRule$ and
rules $\GndDoubleRule$ and $\ParHRule$ to remove
all regular wires, duplicated \ground generators, and parallel Hadamard wires,
ensuring that the final combined diagram is in a weakly graph-like form.
An example of this translation is shown in Figure~\ref{fig:example-into-zx}.

Notice that the translation maps both classical and quantum wires to regular
\zxGND diagram edges.  We keep track of which inputs and outputs of the diagram
were connected to classical wires and introduce \ground generators for the
operations that interact with the environment.
In Section~\ref{sec:extraction} we present a method to detect the sections of
the final circuit that can be implemented as classical operations by looking at
the classical inputs/outputs and the \ground generators, independently of which
wires where originally classical.

\begin{table}[bt]%[H]%
  \begin{spreadlines}{1em}
    \begin{gather*}%
      \input{circuit/cnot.tikz}
      \;\mapsto\;%
      \input{circuit/cnotZX.tikz}
      \qquad
      \input{circuit/hadamard.tikz}
      \;\mapsto\;%
      \input{circuit/hadamardZX.tikz}
      \qquad
      \input{circuit/z.tikz}
      \;\mapsto\;%
      \input{circuit/zZX.tikz}
      \\
      \input{circuit/init.tikz}
      \;\mapsto\;%
      \input{circuit/initZX.tikz}
      \qquad
      \input{circuit/termination.tikz}
      \;\mapsto\;%
      \input{circuit/terminationZX.tikz}
      \qquad
      \input{circuit/xor.tikz}
      \;\mapsto\;%
      \input{circuit/xorZX.tikz}
      \\
      \input{circuit/not.tikz}
      \;\mapsto\;%
      \input{circuit/notZX.tikz}
      \qquad
      \input{circuit/clone.tikz}
      \;\mapsto\;%
      \input{circuit/cloneZX.tikz}
      \qquad
      \input{circuit/meas.tikz}
      \;\mapsto\;%
      \input{circuit/measZX.tikz}
      \qquad
      \input{circuit/prepare.tikz}
      \;\mapsto\;%
      \input{circuit/prepareZX.tikz}
      \\
      \input{circuit/notClassic.tikz}
      \;\mapsto\;%
      \input{circuit/notClassicZX.tikz}
      \qquad
      \input{circuit/zClassic.tikz}
      \;\mapsto\;%
      \input{circuit/zClassicZX.tikz}
      \\
      \input{circuit/wireC.tikz}
      \;\mapsto\;%
      \input{circuit/wireZX.tikz}
      \qquad
      \input{circuit/wireQ.tikz}
      \;\mapsto\;%
      \input{circuit/wireZX.tikz}
      \qquad
      \input{circuit/swapC.tikz}
      \;\mapsto\;%
      \input{circuit/swapZX.tikz}
      \qquad
      \input{circuit/swapQ.tikz}
      \;\mapsto\;%
      \input{circuit/swapZX.tikz}
      \\
      \input{circuit/serialComp.tikz}
      \;\mapsto\;%
      \input{circuit/serialCompZX.tikz}
      \qquad\qquad
      \input{circuit/parallelComp.tikz}
      \;\mapsto\;%
      \input{circuit/parallelCompZX.tikz}
    \end{gather*}
  \end{spreadlines}
  \caption[Translation from hybrid quantum-classical circuits into grounded ZX diagrams.]{%
    Translation from hybrid quantum-classical circuits into \zxGND diagrams.
  }% We cannot use tikz symbols in the caption because it is a moving argument
  \label{table:circ-zx-translation}
\end{table}

\begin{figure}[th]
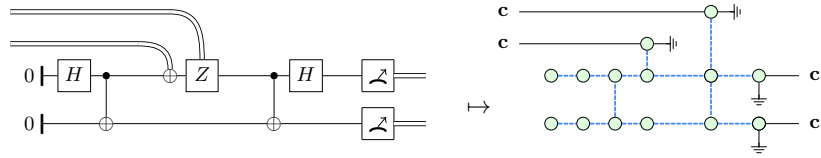

  \[
    \scalebox{0.85}{\input{examples/superdenseCoding-circuit.tikz}}
    %\quad\mapsto\quad
    %\tikzfig{examples/superdenseCoding-zx-full}
    \quad\mapsto
    \scalebox{0.85}{\input{examples/superdenseCoding-zx.tikz}}
  \]
  \caption[% Alt-caption without tikz elements
    Example translation of the superdense coding circuit into a
    zx-ground diagram with labelled inputs and outputs, and subsequent
    application of the spider-fusion rule.
  ]{%
    Example translation of the superdense coding circuit into a
    \zxGND diagram with labelled inputs and outputs, and subsequent
    application of the spider-fusion rule.
  }%
  \label{fig:example-into-zx}
\end{figure}

\begin{lemma}%
  \label{lem:circ-to-graphLike}
  The \zxGND diagram resulting from the translation $\toZxGND{\cdot}$ is weakly
  graph-like.
\end{lemma}

\begin{proof}%
  By induction on the circuit construction. Cf.
  \aplasorlong{long version on arXiv}{Appendix~\ref{sec:appendix-proofs}}.
\end{proof}

After the translation, we can apply Lemma~\ref{lem:weak-to-strict-graphlike} to
obtain a strictly graph-like diagram.
This step essentially separates the \ground generators from the inputs and
outputs, allowing the optimization procedure to move them around and let them
interact with other parts of the diagram.

\begin{lemma}%
  \label{lem:weakGL-gFlow}
  If $C$ is a hybrid quantum-classical circuit and $D$ is the graph-like
  \zxGND diagram obtained from the translation $\toZxGND{C}$ and
  Lemma~\ref{lem:weak-to-strict-graphlike}, then $G(D)$ admits
  a focused gFlow.
\end{lemma}

\begin{proof}%
  By induction on $C$. Cf.
  \aplasorlong{long version on arXiv}{Appendix~\ref{sec:appendix-proofs}}.
\end{proof}

\section{Grounded ZX optimization}%
\label{sec:optimization}

Our simplification strategy for \zxGND\ diagrams is based on eliminating nodes
from the diagram by systematically applying a number of rewriting rules while
preserving the existence of a focused gFlow. In this section we introduce the
new rules, define a strategy to maximize their effectiveness, and finally use
it together with the pure-ZX optimization to define our algorithm.

\subsection{Basic simplification rules}

Duncan et al.~\cite{DKPW2019qcircSimpl} presented the following gFlow-preserving
\textit{local complementation} and  \textit{pivoting} rules for the ZX calculus
in their optimization procedure. These rules effectively reduce the size of the
diagram by at least one node on each application by eliminating internal
proper-Clifford spiders and Pauli spider pairs respectively.
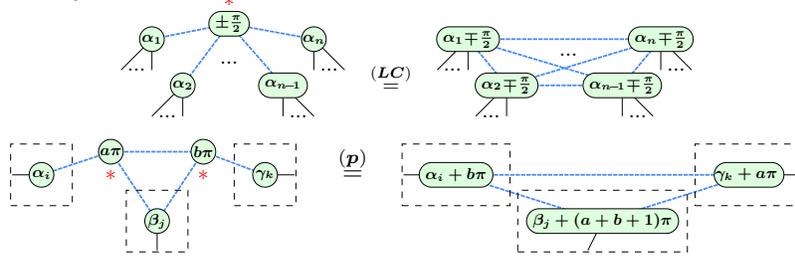
\vspace*{-1em}
\[
  \hfill \scalebox{0.82}{\input{simpl/lc-simp.tikz}} \hfill
\]
\[
  \scalebox{0.82}{\input{simpl/pivot-a.tikz}}
  \quad\overset{\PivotRule}{=}\quad
  \scalebox{0.82}{\input{simpl/pivot-b.tikz}}
\]

These rules can be applied directly in \zxGND\ diagrams when the target spiders
are not connected to a \ground generator. For the cases where some of the target
spiders are \ground-spiders, we introduce the following altered rules.
Their derivation can be found in
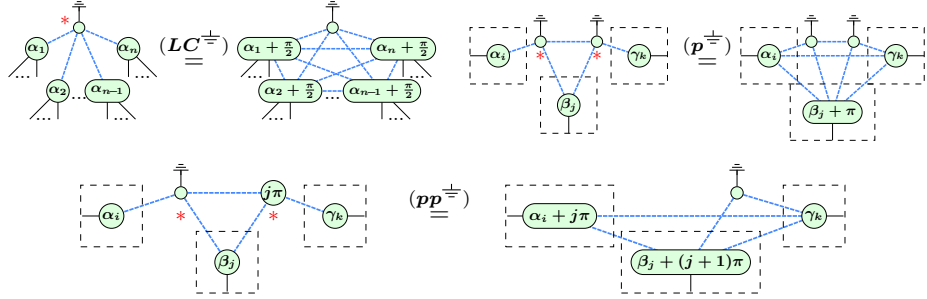
\aplasorlong{the long version on arXiv}{Appendix~\ref{sec:derivations}}.
\[
  \scalebox{0.75}{\input{simpl/gnd-rewriting/gndComplement-a.tikz}}
  \!\!\!\overset{\LocalCompGndRule}{=}
  \scalebox{0.75}{\input{simpl/gnd-rewriting/gndComplement-b.tikz}}
  \quad
  \scalebox{0.75}{\input{simpl/gnd-rewriting/gndPivot-a.tikz}}
  \overset{\PivotGndRule}{=}
  \scalebox{0.75}{\input{simpl/gnd-rewriting/gndPivot-b.tikz}}
\]
\[
  \scalebox{0.82}{\input{simpl/gnd-rewriting/gndDeletion-a.tikz}}
  \quad\overset{\PauliPivotGndRule}{=}\quad
  \scalebox{0.82}{\input{simpl/gnd-rewriting/gndDeletion-b.tikz}}
\]
Notice that both rules $\LocalCompGndRule$ and $\PivotGndRule$ do not decrease
the number of spiders in the diagram. As such, we will focus on rule
$\PauliPivotGndRule$ for our optimization.

If \PauliPivotGndRule\ is applied with a non-\ground spider connected to a
boundary, the rule produces a \ground-spider connected to an input or output
thus needing to add an identity operation as described in
Lemma~\ref{lem:weak-to-strict-graphlike} to preserve the graph-like property.
Since in this case we add additional nodes to the graph, we will only apply rule
\PauliPivotGndRule on a boundary spider if it can be followed by another
node-removing rule.

Additionally, we will use rules $\GndHDiscardRule$ and $\GndRemoveRule$ directly
to remove nodes in the diagram when there are \ground-spiders with degree 1 or
0\ in the graph, respectively.

\begin{lemma}
  If the non-\ground spider in the lhs of the discarding rule~\GndHDiscardRule\
  is not connected to an output or input, then applying the rule over a
  graph-like diagram $D$ preserves the existence of a focused gFlow.
\end{lemma}

\begin{proof}
  If the non-\ground spider in the lhs is not connected to an input or output
  of the diagram, then applying the rule does not break the graph-like
  property of $D$. The preservation of the gFlow follows from \ground-spiders
  being sinks of the underlying open graph.
\end{proof}

\begin{lemma}
  Rules~\LocalCompGndRule, \PivotGndRule, and~\PauliPivotGndRule\ preserve the
  existence of a focused gFlow.
\end{lemma}

\begin{proof}
  Notice that that rules~\LocalCompGndRule, \PivotGndRule, 
  and~\PauliPivotGndRule\ are compositions of gFlow-preserving rules.
\end{proof}

\subsection{Ground-cut simplification}%
\label{sec:ground-simplification}

The previously introduced rewriting rules require a simplification strategy to
apply them. A simple solution is to try to find a match for each rule
and apply them iteratively until no more matches are available. We describe a
strategy that can find additional rule matches by operating on the biadjacency
matrix between the \ground-spiders and the non-\ground spiders.
\begin{definition}%
  \label{defn:ground-cut}
  The \textit{ground-cut} of a graph-like \zxGND diagram $D$
  is the cut resulting from splitting the \ground and non-\ground spiders in $G(D)$.
\end{definition}

Since the diagram is graph-like, there are no internal wires in the \ground
partition.
Given a \zxGND diagram $D$, we denote $M_D$ the biadjacency matrix of its
ground-cut, where rows correspond to \ground-spiders and columns correspond to
non-\ground spiders.
We can apply all elementary row operations on the matrix by rewriting the
diagram. The addition operation between the rows corresponding to the
\ground-spider $u$ and the \ground-spider $v$ can be implemented via the
following rule\aplasorlong{}{, the derivation of which can be found in
Appendix~\ref{sec:derivations}}.
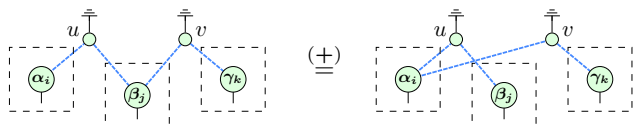
\[
  \scalebox{0.85}{\input{simpl/gnd-row-sum-0.tikz}}
  \quad\overset{\RowSumGndRule}{=}\quad
  \scalebox{0.85}{\input{simpl/gnd-row-sum-4.tikz}}
\]

Using the elementary row operations we can apply Gaussian elimination on the
ground-cut biadjacency matrix of a graph-like \zxGND\ diagram, generating in the
process an equivalent diagram whose ground-cut biadjacency matrix is in
reduced echelon form.
%The elimination process runs in $\bigo(n^2*m)$ steps, where $n$ is the number
%of ground generators and $m$ is the number of spiders adjacent to a
%\ground-spider~\citationneeded.

Any row in the ground-cut biadjacency matrix left without non-zero elements
after applying Gaussian elimination corresponds to an isolated
\ground-spiders in the diagram that can be eliminated by rule \GndRemoveRule.
If the reduced row echelon form of the biadjacency matrix contains a row with
exactly one non-zero elements, then that element corresponds to an isolated
\ground-spider and non-\ground spider pair in the diagram and therefore we can
apply rule \GndHDiscardRule\ to remove the non-\ground spider.
%
%After the reduction, the biadjacency matrix of the ground-cut of the resulting
%diagram can again be reduced to row echelon form in $\bigo(n*m)$ steps where $n$
%is the number of ground generators and $m$ is the number of nodes adjacent to a
%\ground-spider.

\subsection{The Algorithm}%
\label{sec:algorithm}

Based on the previous strategy, we define a terminating procedure which turns
any graph-like \zxGND diagram into an equivalent \textit{simplified} diagram
that cannot be further reduced.

\begin{definition}
  A graph-like \zxGND diagram is in simplified-form if it does not contain any
  of the following, up to single-qubit unitaries on the inputs and outputs.
  \begin{itemize}
    \item Interior proper Clifford spiders.
    \item Adjacent pairs of interior Pauli spiders.
    \item Interior Pauli spiders adjacent to boundary spiders.
    \item Interior Pauli spiders adjacent to \ground-spiders.
    \item Degree-1 \ground-spiders not connected to input or output spiders.
    \item Connected components not containing inputs nor outputs.
  \end{itemize}
\end{definition}
\ifaplas
\else
  An example of diagrams satisfying and not satisfying this property is shown
  in Figure~\ref{fig:simplified-example}

  \begin{figure}[tb]
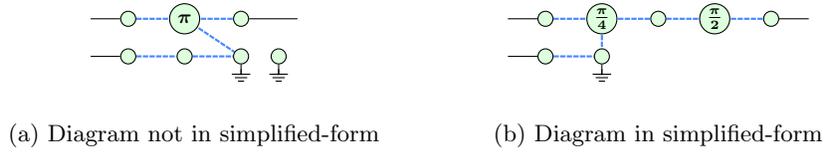
%
    \begin{subfigure}[b]{0.5\textwidth}
         \centering
         \[\input{examples/simplified-not.tikz}\]
         \caption{Diagram not in simplified-form}
     \end{subfigure}
     %\hfill
     \begin{subfigure}[b]{0.5\textwidth}
         \centering
         \[\input{examples/simplified.tikz}\]
         \caption{Diagram in simplified-form}
     \end{subfigure}
    \caption{Example of simplified and non-simplified diagrams.}%
    \label{fig:simplified-example}%
  \end{figure}
\fi

We define an optimization algorithm that produces diagrams in simplified-form by
piggybacking on the pure optimization procedure. This optimization applies the
local complementations \LocalCompRule\ and pivoting \PivotRule\ rules until
there are no interior proper Clifford spiders or adjacent pairs of non-\ground
interior Pauli spiders.
After the initial pure simplification, we continue our optimization as follows.
\begin{enumerate}
  \item \label{algorithm:ground-loop}
    Repeat until no rule matches, removing wires between \ground-spiders and
    parallel Hadamard connections after each step:
    \begin{enumerate}
      \item Run Gaussian elimination on the ground-cut of the diagram
        as described in Section~\ref{sec:ground-simplification}.
      \item \label{algorithm:reduction-step-disconnected}
        Remove the grounds corresponding to null rows with rule \GndRemoveRule.
      \item \label{algorithm:reduction-step-discard}
        If any row of the biadjacency matrix has a single non-zero element,
        corresponding to a \ground-spider connected to a spider $v$, then:
        \begin{enumerate}
          \item If $v$ is not a boundary spider, apply rule~\GndHDiscardRule.
          \item If $v$ is a boundary spider and $v$ is adjacent to a Pauli spider,
            apply rule~\GndHDiscardRule\ immediately followed by
            the procedure from Lemma~\ref{lem:weak-to-strict-graphlike}
            to make the diagram graph-like again.
            Then delete the Pauli neighbour using rule~\PauliPivotGndRule,
            to ensure that the step removes at least one node.
        \end{enumerate}
      \item Apply Pauli spider elimination rule~\PauliPivotGndRule\ 
        until there are no Pauli spiders connected to ground spiders.
    \end{enumerate}
  \item Remove any connected component of the graph without inputs or outputs.
\end{enumerate}

Notice that each cycle the loop reduces the number of nodes in the graph, so
this is a terminating procedure. Additionally, since each applied rule preserves
the existence of a gFlow the final diagram admits a gFlow. An example run of
the algorithm is shown in Figure~\ref{fig:optimization}.

\begin{figure}[tb]
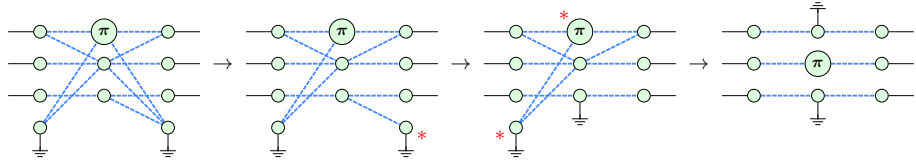
%
  \noindent
  \centering
  \scalebox{0.85}{%
    \(
      \input{examples/optimizationB-0.tikz}
      \;\to\;
      \input{examples/optimizationB-1.tikz}
      \;\to\;
      \input{examples/optimizationB-2.tikz}
      \;\to\;
      \input{examples/optimizationB-3.tikz}
    \)
  }%
  \caption{Example of a diagram optimization applying a ground-cut
  simplification, a discard rule, and a Pauli elimination.}%
  \label{fig:optimization}%
\end{figure}

\section{Circuit extraction}%
\label{sec:extraction}

Here we describe a general circuit extraction procedure for graph-like
\zxGND diagrams admitting a focused gFlow into hybrid quantum-classical
circuits, by modifying the procedure for pure diagrams from the Clifford
optimization. We present the pseudocode in
Algorithm~\ref{algo:extraction}\aplasorlong{}{ and an example of an execution in
Appendix~\ref{sec:extraction-example}}.

The algorithm progresses through the diagram from right-to-left, maintaining a
set of spiders $\frontier$, called the \textit{frontier}, which represents the
unextracted spiders connected to the extracted segment. Each frontier spider is
assigned an output qubit line $\fQubit(v)$. This set is initially populated by
the nodes connected to outputs of the diagram. The strategy is to proceed
backwards by steps, adding unextracted spiders into the frontier and deleting
some of them to extract operations on the output circuit, in back-to-front
order.

To find candidate spiders to add to the frontier we apply Gaussian elimination
on the biadjacency matrix of the frontier and non-frontier spiders, similarly to
the optimization method described in Section~\ref{sec:ground-simplification}.
The gFlow property of the graph ensures that we can always progress by
extracting a node.  It suffices to look at the set of non-frontier vertices
maximal in the order and notice that, after the Gauss elimination, either we can
choose a \ground-spider from the set, or a non-\ground spider that has a single
connection to the frontier. A careful implementation of the biadjacency matrix
row and column ordering can reduce the number of \ground-spider extractions when
no non-\ground candidates are available. We require the following proposition
to apply the row additions on the graph
(Duncan et al.~\cite{DKPW2019qcircSimpl}, Proposition 7.1).

\begin{proposition}%
  \label{lem:row-sum-extraction}
  For any \zxGND diagram $D$, the following equation holds:
  \[
    \input{extraction-v2/row-sum-0.tikz}
    \quad=\quad
    \input{extraction-v2/row-sum-1.tikz}
  \]
  where $M$ describes the biadjacency matrix of the relevant vertices, and
  $M'$ is the matrix produced by adding row 2 to row 1 in $M$.
  Furthermore, if the diagram on the LHS has a focused gFlow, then so does the
  RHS.
\end{proposition}

In our pseudocode, the call to \textsc{CleanFrontier} ensures that $\frontier$
only contains phaseless spiders without internal wires. Notice that it preserves
the gFlow since it only modifies edges between sink nodes, and removes spiders
with no other connections.
After the while loop terminates, all outputs of the circuit will have been
extracted. If there are inputs left unextracted, and since the diagram had a
gFlow, we can discard them directly via measurement operations.

Finally, we add any necessary swap operations to map the inputs to the
correct lines, and insert qubit initializations and measurements at
inputs and outputs marked as classical. In Section~\ref{sec:classicalization}
we detail a method to better detect the internal parts of the circuit that can
be implemented classically.

In any case, each step of the while loop in
Algorithm~\ref{algo:extraction} line~\ref{algo:extraction:while},
preserves the gFlow of the diagram, and we can show that it terminates
in at most $|V|$ steps: Indeed, if there are no non-frontier spiders,
then a call to \textsc{CleanFrontier} will remove all nodes from the
frontier. Moreover, each step of the while loop in
line~\ref{algo:extraction:while} moves one non-frontier spider to
$\frontier$. 

% \begin{lemma}%
%   \label{lem:extraction-gflow}
%   Each step of the while loop in Algorithm~\ref{algo:extraction},
%   line~\ref{algo:extraction:while}, preserves the gFlow of the diagram.
% \end{lemma}

% \begin{proof}
%   By looking at the set of non-frontier spiders maximal in the gFlow order. Cf.
%   \aplasorlong{long version on arXiv}{Appendix~\ref{sec:appendix-proofs}}.
% \end{proof}

% \begin{lemma}
%   The algorithm terminates.
% \end{lemma}

% \begin{proof}
%   Notice that if there are no non-frontier spiders, then a call to
%   \textsc{CleanFrontier} will remove all nodes from the frontier. Notice that
%   each step of the while loop in line~\ref{algo:extraction:while} moves one
%   non-frontier spider to $\frontier$. Therefore, the loop terminates in at most
%   $|V|$ steps.
% \end{proof}

\begin{algorithm}[t]
\caption{Circuit extraction}\label{algo:extraction}
\begin{algorithmic}[1]
\Function{Extraction}{}
  \State $\frontier : Set\langle Node\rangle \gets \outputs$,
    $\fQubit : Map\langle Node, int\rangle \gets \emptyset$
  \ForAll{$v \in \frontier$}
    $\fQubit(v) \gets$ Output connected to $v$
  \EndFor
  \State \textsc{CleanFrontier}$(\frontier, \fQubit)$
  \While{$\frontier \neq \emptyset$} \label{algo:extraction:while}
    \State Run Gauss elimination on the frontier biadjacency matrix $M$
    (Proposition~\ref{lem:row-sum-extraction})
    %\Statex \Comment $M$ is in row echelon form
    \If{a row of $M$ has a single non-zero element}
      \State Let $u$ and $v$ be the corresponding non-frontier and frontier node
      \State $\fQubit(u) \gets \fQubit(v)$
      \State Remove $v$ from the diagram and add $u$ to $\frontier$
    \Else
      \State $v \gets$ Arbitrary \ground-spider in the neighbourhood of $\frontier$
      \State $\fQubit(v) \gets$ New qubit line id
      \State Extract a classical bit termination on $\fQubit(v)$ and
        add $u$ to $\frontier$
    \EndIf
    \State \textsc{CleanFrontier}$(\frontier, \fQubit)$
  \EndWhile
  \ForAll{Unextracted $v \in I$}
      \State $\fQubit(v) \gets$ Input connected to $v$
      \State Extract a measurement gate
        and a classical bit termination on $\fQubit(v)$
      \State Assign the corresponding input to $\fQubit(v)$
  \EndFor
\EndFunction

\Function{CleanFrontier}{$\frontier, \fQubit$}
  \ForAll{$v \in \frontier$}
    \If{$v$ is a \ground-spider}
      Remove the \ground, extract a measurement on $\fQubit(v)$
    \EndIf
    \If{$v$ has a phase $\alpha \neq 0$}
      Set $\alpha = 0$, extract a \Z[\alpha] gate on $\fQubit(v)$
    \EndIf
    \ForAll{$u\in\frontier, v \sim u$}
      Remove the wire, extract a \CZ\ gate on $\fQubit(v)$, $\fQubit(u)$
    \EndFor
    \If{$v$ is not connected to any other node}
      \State Remove $v$
      \If{$v \in \inputs$}
        assign the input to qubit $\fQubit$
      \Else{}
        extract a $\ket{+}$ qubit initialization on $\fQubit(v)$
      \EndIf
    \EndIf
  \EndFor
\EndFunction
\end{algorithmic}
\end{algorithm}

\section{Circuit classicalization}%
\label{sec:classicalization}

The extraction procedure described in Section~\ref{sec:extraction} produces
correct circuits that are almost completely composed by quantum gates and wires,
without any classical operation.
In this section we describe the general problem of detecting parts of the
circuits that can be realized as classical operations, and introduce an
efficient heuristic solution based on a local-search.
Notice, however, that while we aim to recognize all classically realizable
operations in the circuit the characteristics of each quantum computer may
dictate the final choice between quantum and classical operators by taking into
account the costs of exchanging data between both realms.

Given a \zxGND diagram, we decorate its wires using the set of labels $\labels =
\{\Q, \X, \Y, \Z, \bot\}$. The label $\Z$ means that this particular wire can be
replaced by a classical wire (possibly precomposed with a standard basis
measurement and postcomposed by a qubit initialisation in the standard basis
depending on whether the connected wires are also classical or not), and similarly for
$\X$ and $\Y$ by adapting the basis of measurement/initialisation to the diagonal
or $Y$ basis. $\Q$ means that the wire is a quantum wire, and finally $\bot$
means that the wire can be removed by precomposing with a \ground and postcomposing
with a maximally mixed state.  The set of labels form a partial order, $\Q \ge
\X,\Y,\Z \ge \bot$.
%\[ \tikzfig{classical/label-order} \]

A labelling $L$ of a diagram $D$ is a map from its edges into a pair of labels.
The two labels, drawn at each end of the wire, indicate the origin of the
constraint. Intuitively, $\input{classical/id-za.tikz}$ means that the wire is
produced in such a way that guarantees that the qubit carries classical
information encoded in the computational basis, whereas
$\input{classical/id-az.tikz}$ means that the wire can be replaced by a classical
wire because some process will force this qubit to be in that basis ---for
instance, it is going to be measured in the standard basis and thus one can
already measure this qubit in the standard basis and use a classical wire---.
We define a partial order between labellings of a diagram as the natural lift
from the partial order of the labels.

Each label corresponds to a density matrix subspace of $\C^{2\times2}$,
representing all possible mixed states allowed by that particular kind of wire.
\[
\begin{array}{ll}
  \interpretLbl{\Q} = \C^{2\times2}
  &
  \interpretLbl{\Z} =%
    \{\alpha \ketbra{0}{0} + \beta \ketbra{0}{0} \mid%
    \alpha,\beta\in\R_{\geq0}, \alpha+\beta=1\}
  \\
  \interpretLbl{\bot} = \{\frac12 \ketbra00 + \frac12 \ketbra11\}
  \quad
  &
  \interpretLbl{\X} =%
    \{\alpha \ketbra{+}{+} + \beta \ketbra{-}{-} \mid%
    \alpha,\beta\in\R_{\geq0}, \alpha+\beta=1\}
  \\
  &
  \interpretLbl{\Y} =%
    \{\alpha \ketbra{\circlearrowright}{\circlearrowright} + \beta
      \ketbra{\circlearrowleft}{\circlearrowleft} \mid%
    \alpha,\beta\in\R_{\geq0}, \alpha+\beta=1\}
\end{array}%
\]
Notice that the greatest common ancestor $\interpretLbl{A \sqcup B}$ corresponds
to the intersection of the sets.

\begin{wrapfigure}[8]{R}{2in}
  \vspace{-3ex}{\input{classical/validity-.tikz}}
\end{wrapfigure}
Intuitively, a labelling is valid if we can cut any wire in the diagram and,
after forcing a valid state in the inputs and outputs, we get a valid state in
the cut terminals. That is, we rearrange the diagram to transform all outputs
into inputs and connect the cut terminals as outputs, as shown on the right.
Then, applying an arbitrary input
\(
  \rho \in (\bigotimes_i^n \interpretLbl{A_i})
    \otimes(\bigotimes_j^m \interpretLbl{D_j})
\)
to the diagrams produces a result in $\interpretLbl{E}\otimes\interpretLbl{F}$.

Notice that if $A$ is a valid label for a wire then any $B \geq A$ is also
valid, and in particular \Q\ is always a valid label. We can then omit
unnecessary labels in the diagrams, marking them implicitly as \Q.

Given a \zxGND diagram $D$ with marked classical inputs and outputs,
we define the classicalization problem as finding a minimal valid labeling
where the inputs and outputs are labelled as \Q\ or \Z\ accordingly.

\subsection{Local-search algorithm}

We present a local-search labelling procedure for \zxGND\ diagrams with explicit
Hadamard gates ---replacing the Hadamard wires--- and only green spiders, that
produces locally minimal labellings by propagating the labels over individual
spiders. 
A diagram resulting from the circuit extraction in Section~\ref{sec:extraction}
can be transformed to have only green spiders by applying the color-change rule
\HadamardRule. This restriction is purely for simplicity in our definition, as
the equivalent functions can be defined easily for red spiders.

We introduce a number of operations over the labels.
First, a binary function representing the result of combining two wires via a
phaseless green spider, $\tJoin:\labels\times\labels\to\labels$.
\[
\begin{array}{lllll}
  \Z\tJoin A=Z \quad & \X\tJoin A=A \quad & \Y\tJoin\Y=\X \quad & \Q\tJoin\Y = \Q \quad & \bot\tJoin\Y = \bot \\
  A\tJoin\Z=Z & A\tJoin\X=A & \Y\tJoin\Q=\Q & \Q\tJoin\Q=\Q & \bot\tJoin\Q=\Z \\
  & & \Y\tJoin\bot=\bot & \Q\tJoin\bot=\Z & \bot\tJoin\bot=\bot \\
\end{array}
\]
Notice that $(\labels, \tJoin)$ is a commutative monoid with \X\ as neutral
element.
We also define a ``Z rotation" operation for $\alpha\in[0,2\pi)$,
$\rot: \labels \to \labels$.
\[
\begin{array}{llllll}
  \multicolumn{2}{l}{\rot(\Z) = \Z \qquad\qquad} &
  \multicolumn{2}{l}{\rot(\Q) = \Q \qquad\qquad} &
  \multicolumn{2}{l}{\rot(\bot) = \bot} \\
  \multicolumn{3}{l}{%
    \rot(\X) =
    \begin{cases}
      \X & \text{if $\alpha \in \{0,\pi\}$} \\
      \Y & \text{if $\alpha \in \{\frac12\pi,\frac34\pi\}$} \\
      \Q & \text{otherwise}
    \end{cases}
    \qquad
  } &
  \multicolumn{3}{l}{%
    \rot(\Y) =
    \begin{cases}
      \Y & \text{if $\alpha \in \{0,\pi\}$} \\
      \X & \text{if $\alpha \in \{\frac12\pi,\frac34\pi\}$} \\
      \Q & \text{otherwise}
    \end{cases}
  }
\end{array}
\]
This corresponds to the identity if $\alpha\in\{0,\pi\}$, and in general
\(
  \rot(A) \tJoin b \geq \rot(A \tJoin B)
\).

Finally, we define a function $\hadamLbl$ representing the application of the
Hadamard operation over a label, $\hadamLbl: \labels \to \labels$.
\[
  \hadamLbl(\Q) = \Q
  \qquad
  \hadamLbl(\X) = \Z
  \qquad
  \hadamLbl(\Z) = \X
  \qquad
  \hadamLbl(\Y) = \Y
  \qquad
  \hadamLbl(\bot) = \bot
\]

Our classical detection procedure starts by labelling any classical input or
output with a $\Z$ label, and any \ground with a $\bot$ label, and the rest of
the diagram wires with $\Q$.

It then proceeds by propagating the labels using the
following rules:
\[
  \input{classical/had-rule0.tikz}
  \;\;\overset{\mbox{\scriptsize\ClassHadamRule}}{\to}\;\;
   \input{classical/had-rule1.tikz}
\]
\[
  \input{classical/Zn-p-rule0.tikz}
  \;\;\overset{\mbox{\scriptsize\ClassZRule}}{\to}\;\;
  \input{classical/Zn-p-rule1.tikz}
\]
For any labels $A,B,C,D,E,F\in\labels$.

We apply these rules until there are no more labels to change. Since each time we
replace labels with lesser ones in the order, the procedure terminates.
Finally, we can interpret wires with a classical label in any direction as
classically realisable.
%possibly adding Hadamard operations or swapping the color of spiders to change
%the basis for $\X$ and $\Y$ wires.
We show an example of a labeled diagram in
Figure~\ref{fig:classicalization-example}.

\begin{figure}[ht]
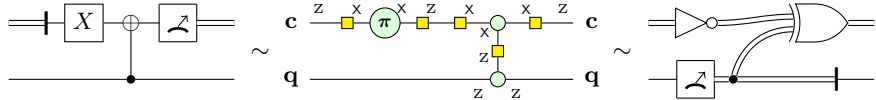
%
  \noindent
  \[
    \input{examples/classicalization-start-circ.tikz}
    \;\sim\;
    %\tikzfig{examples/classicalization-start}
    %\;\xrightarrow{Labelling}\;
    \input{examples/classicalization-labelled.tikz}
    %\;\xrightarrow{\substack{Color\\change}}\;
    %\tikzfig{examples/classicalization-colorChange}
    \;\sim\;
    \input{examples/classicalization-labelled-circ.tikz}
  \]
  \caption{Example of a local-search classicalization.
    \Q\ labels are omitted.}%
  \label{fig:classicalization-example}%
\end{figure}

\begin{lemma}%
  \label{lem:localClassicalization}
  The local-search labelling algorithm produces a valid labeling according to
  the standard interpretation of the \zxGND\ calculus.
\end{lemma}

\begin{proof}
  By proving that both rules \ClassHadamRule\ and \ClassZRule\ produce valid
  labellings, cf.
  \aplasorlong{long version on arXiv}{Appendix~\ref{sec:appendix-proofs}}.
\end{proof}

\ifaplas
\else
  \begin{remark}
    The local-search labelling does not always reach a globally minimum labeling.
    Consider for example the following input diagram.
    \[
    \input{classical/cnot-cnot-.tikz}
    \]
    Here, although the diagram is equivalent to the identity and the outputs could
    theoretically be labelled as $\Z$, the entangled wires in the middle do not
    allow for local propagation of the classical labels past the double controlled
    Z gates.
  \end{remark}
\fi

\section{Implementation}%
\label{sec:implementation}

We have implemented each of the algorithms presented in this work as an
extension to the open source Python library
\textit{PyZX}~\cite{kissinger_pyzx_2020} by modifying its implementation of ZX
diagrams to admit \zxGND\ primitives. A repository with the code is
available at \url{http://github.com/aborgna/pyzx/tree/zxgnd}.
We additionally implemented a naïve \zxGND\ extension of the pure Clifford
optimization for comparison purposes, which doesn't use any of our \ground
rewriting rules.
When applied to pure quantum circuits, our algorithm does not perform additional
optimizations after the Clifford procedure and therefore achieves the same
benchmark results recorded by Duncan et al.\ on the circuit set described by Amy et
al.~\cite{amy_polynomial-time_2014}.

We tested the procedure over two classes of randomly generated circuits, and
measured the size of the resulting diagram as the number of spiders left after
the optimization. This metric correlates with the size of the final circuit,
although the algorithmic noise caused by the arbitrary choices in the extraction
procedure may result in some cases in bigger extracted circuit after a reduction
step.

The first test generates Clifford+T circuits with measurements by applying randomly
chosen gates from the set \{\CNOT, S, HSH, T, Meas\} over a fixed number of
qubits, where Meas are measurement gates on a qubit immediately followed by a
qubit initialization.  We fix the probability of choosing a \CNOT, S, or HSH
gate to $0.2$ each and vary the probabilities for T and Meas in the remaining
$0.4$. These circuits present a general worst case, where there is no
additional classical structure to exploit during the hybrid circuit
optimization. 

The second type of generated operations are classical parity-logic circuits.
These consist on a number of classical inputs, fixed at 10, where we apply
randomly chosen operations from the set \{NOT, XOR, Fanout\} with probabilities
0.3, 0.3, and 0.4 respectively.

In Figure~\ref{fig:benchmarks} we compare the results of our optimization using
the Clifford optimization as baseline.
Figure~\ref{fig:benchmarks:clifford} shows the reduction of diagram size when
running the algorithm on randomly generated Clifford+T circuits with
measurement. We vary the probability of generating a measurement gate between
$0$ and $0.2$ while correspondingly changing the probability of generating a
T-gate between $0.4$ and $0.2$, and show the results for different combinations
of qubit and gate quantities.
We remark that the optimization produces noticeably smaller diagrams once enough
\ground generators start interacting with each other. There is a critical
threshold of measurement gate probability, specially visible in the cases with 8
qubits and 1024 gates, where with high probability the outputs of the diagram become
disconnected from the inputs due to the \ground interactions. This results in
our algorithm optimizing the circuits to produce a constant result while
discarding their input.

Figure~\ref{fig:benchmarks:parity} shows the comparison of diagram size between
our procedure and the Clifford optimization when run over classical parity
circuits.
The optimization produces consistently smaller diagrams, generally
achieving the theoretical minimal number of \ground generators, equal to the
number of inputs.
We further remark that in all of the tested cases the classicalization procedure
was able to detect that all the extracted operations on the optimized
parity-logic circuit were classically realisable.

\begin{figure}[tb]
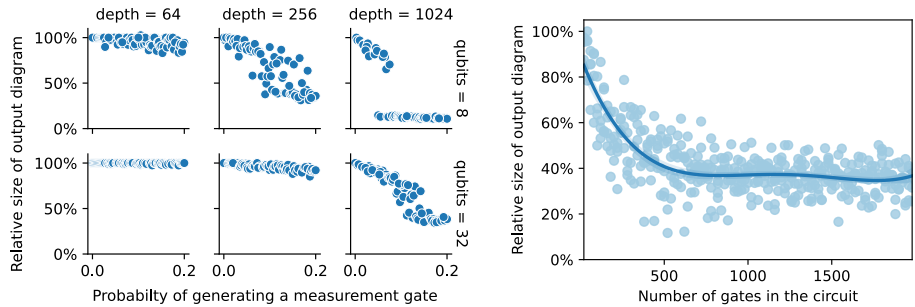
%
  % !!! Compile with LaTeX flag '--shell-escape' to regenerate the images
  \begin{subfigure}[t]{0.5\textwidth}
       \includesvg[height=13em]%
         {benchmarks/sizeReduction-measurementProbability-cliffordT-alternative.svg}
       \caption{Diagram size reduction on Clifford+T circuits with
         measurements.}%
       \label{fig:benchmarks:clifford}
       \hfill
   \end{subfigure}
   \begin{subfigure}[t]{0.5\textwidth}
       \hfill
       \includesvg[height=12.5em]%
         {benchmarks/sizeReduction-nGates-parityCirc.svg}
       \caption{Diagram size reduction on parity-logic circuits.}%
       \label{fig:benchmarks:parity}
   \end{subfigure}
  \vspace*{-1em}
  \caption{Benchmark results on randomly generated diagrams.}%
  \label{fig:benchmarks}%
\end{figure}

The runtime of our algorithm implementation is polynomial in the size of the
circuit. As with the Clifford optimization, the cost of our optimization and
extraction processes is dominated by the Gauss elimination steps.
For the ground-node rewriting rules, our unoptimized implementation is roughly
$\bigo(n^2 * k^2)$ in the worst case with $k$ being the number of measurement gates and $n$ the number of gates,
but in practice it behaves cubically on the number of gates due to the sparseness of the diagrams.
The implementation was not developed with a focus on the runtime cost, and some possible optimizations may reduce this bound.

\section{Discussion and future work}%
\label{sec:discussion}

We introduced an optimization procedure for optimizing hybrid classical circuits
inspired by previous work on pure circuit optimization using the ZX calculus.
The process is composed by a translation step, the optimization of the diagrams,
an extraction back into circuits and finally a detection of
classically-realisable operations.
Our translation operation produces diagrams which admit a focused gFlow, a
property that we maintain during the optimization and require during the
extraction.
For our optimization step we defined a series of rewrite strategies to reduce
the size of the diagrams, and introduced a strategy to find additional
optimization opportunities by applying Gaussian elimination on the biadjacency
matrix of the ground-cut of the diagram. Our extraction procedure initially
generates circuits without classical operations. Hence, we introduced a
classicalization heuristic for arbitrary circuits that is able to replace
quantum operations by their classical equivalent, where possible.

Kissinger and van de Wetering~\cite{aleks_kissinger_reducing_2020} defined a
procedure based on the Clifford optimization to reduce the T-gate count in
quantum circuits by defining new structures in the graphs called \textit{phase
gadgets} and operating over their phases. Their work can be easily extended to
\zxGND, where the \ground generators act as an absorbing element for the
gadgets phases. However, rules such as \RowSumGndRule\ prove to be strictly
more powerful than applying the pure phase gadget rules over \ground-gadgets.
In general, the phase-gadget optimization affects an independent section of the
structure of the diagram compared to ours, and can be applied  with it.

During our definition of the optimization process we decided to restrict the
input circuits to parity classical logic, excluding \AND\ and \OR\ gates. This
does not raise from an inherent limitation of the system but from a practical
standpoint.  The \zxGND\ calculus is able to represent \AND\ operations in what
equates to the Clifford+T decomposition of the Toffoli gate, introducing
multiple T-gates and \CNOT\ gates to the circuit~\cite{Selinger_2013}. The
multiple spiders would be dispersed around the diagram during the optimization
step, potentially breaking the pattern formed by the \AND\ gate and replacing it
with multiple quantum operations. This can produce the unexpected result of
introducing expensive quantum operations in an originally pure classical
circuit.
A possible next step for this work would be to use an alternative diagrammatic
representation better adapted to represent arbitrary boolean circuit such as the
ZH calculus.

\paragraph{Acknowledgements}
{\small The authors would like to thank Kostia Chardonnet and Renaud Vilmart for their
suggestions on the classicalization problem, and John van de Wetering for his
help with the pyzx library.
%f
This work was supported in part by the French National Research Agency (ANR)
under the research project SoftQPRO ANR-17-CE25-0009-02, and by the DGE of the
French Ministry of Industry under the research project PIA-GDN/QuantEx
P163746-484124, and by the project STIC-AmSud project Qapla' 21-SITC-10.}

%\begin{remark}%
%  \label{rem:ground-teleportation}
%  The following \textit{ground teleportation} rule can also be derived
%  from rule~\PauliPivotGndRule,
%  \[
%    \tikzfig{simpl/gnd-rewriting/gndTeleport}
%  \]
%  While the rule does not remove any spider,
%  if there is only one $\beta$ spider then we can use the discarding
%  rule~\GndHDiscardRule\ followed by rule~\PauliPivotGndRule\ 
%  to delete two nodes from the graph.
%  \TODO{Show example of when is this useful}
%\end{remark}

%
% ---- Bibliography ----
%
% BibTeX users should specify bibliography style 'splncs04'.
% References will then be sorted and formatted in the correct style.
%
% \bibliographystyle{splncs04}
% \bibliography{paper}

\ifaplas
\else

  \newpage
  \appendix

  \input{appendix/semantics.tex}

  \input{appendix/derivations.tex}

  \input{appendix/extraction.tex}

\fi

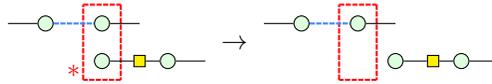
\end{document}

%% file: tikzstyles.tex
\tikzstyle{env}=[copoint,regular polygon rotate=0,minimum width=0.2cm, fill=black]

\tikzstyle{probs}=[shape=semicircle,fill=white,draw=black,shape border rotate=180,minimum width=1.2cm]

\tikzstyle{every picture}=[baseline=-0.25em,scale=0.5]
\tikzstyle{dotpic}=[] % for backwards-compatibility
\tikzstyle{diredges}=[every to/.style={diredge}]
\tikzstyle{math matrix}=[matrix of math nodes,left delimiter=(,right delimiter=),inner sep=2pt,column sep=1em,row sep=0.5em,nodes={inner sep=0pt},text height=1.5ex, text depth=0.25ex]

\tikzstyle{gs edge}=[]
\tikzstyle{gs double edge}=[double,shorten <=-1mm,shorten >=-1mm,double distance=2pt]

\tikzstyle{inline text}=[text height=1.5ex, text depth=0.25ex,yshift=0.5mm]
\tikzstyle{label}=[font=\footnotesize,text height=1.5ex, text depth=0.25ex,yshift=0.5mm]
\tikzstyle{left label}=[label,anchor=east,xshift=1.5mm]
\tikzstyle{right label}=[label,anchor=west,xshift=-1.5mm]

\tikzstyle{braceedge}=[decorate,decoration={brace,amplitude=2mm,raise=-1mm}]
\tikzstyle{small braceedge}=[decorate,decoration={brace,amplitude=1mm,raise=-1mm}]

\tikzstyle{doubled}=[line width=1.6pt] % set the line width for all doubled (quantum) maps/wires
\tikzstyle{boldedge}=[doubled,shorten <=-0.17mm,shorten >=-0.17mm]
\tikzstyle{boldedgegray}=[doubled,gray,shorten <=-0.17mm,shorten >=-0.17mm]
\tikzstyle{singleedgegray}=[gray]%,shorten <=-0.1mm,shorten >=-0.1mm]

\tikzstyle{semidoubled}=[line width=1.4pt] % set the line width for all doubled (quantum) maps/wires
\tikzstyle{semiboldedgegray}=[semidoubled,gray,shorten <=-0.17mm,shorten >=-0.17mm]

\tikzstyle{boxedge}=[semiboldedgegray]

\tikzstyle{boldedgedashed}=[very thick,dashed,shorten <=-0.17mm,shorten >=-0.17mm]
\tikzstyle{vboldedgedashed}=[doubled,dashed,shorten <=-0.17mm,shorten >=-0.17mm]
\tikzstyle{left hook arrow}=[left hook-latex]
\tikzstyle{right hook arrow}=[right hook-latex]
\tikzstyle{sembracket}=[line width=0.5pt,shorten <=-0.07mm,shorten >=-0.07mm]

\tikzstyle{causal edge}=[->,thick,gray]
\tikzstyle{causal nondir}=[thick,gray]
\tikzstyle{timeline}=[thick,gray, dashed]

\tikzstyle{cedge}=[<->,thick,gray!70!white]

\tikzstyle{empty diagram}=[draw=gray!40!white,dashed,shape=rectangle,minimum width=1cm,minimum height=1cm]
\tikzstyle{empty diagram small}=[draw=gray!50!white,dashed,shape=rectangle,minimum width=0.6cm,minimum height=0.5cm]

\tikzstyle{dot}=[inner sep=0mm,minimum width=2mm,minimum height=2mm,draw,shape=circle]  
\tikzstyle{Wsquare}=[white dot, shape=regular polygon, rounded corners=0.8 mm, minimum size=3.3 mm, regular polygon sides=3, outer sep=-0.2mm]
\tikzstyle{Wsquareadj}=[white dot, shape=regular polygon, rounded corners=0.8 mm, minimum size=3.3 mm, regular polygon sides=3, outer sep=-0.2mm, regular polygon rotate=180]
\tikzstyle{ddot}=[inner sep=0mm, doubled, minimum width=2.5mm,minimum height=2.5mm,draw,shape=circle]

\tikzstyle{black dot}=[dot,fill=black]
\tikzstyle{white dot}=[dot,fill=white,,text depth=-0.2mm]
\tikzstyle{white Wsquare}=[Wsquare,fill=white,,text depth=-0.2mm]
\tikzstyle{white Wsquareadj}=[Wsquareadj,fill=white,,text depth=-0.2mm]
\tikzstyle{green dot}=[white dot] % for backwards-compatibility
\tikzstyle{gray dot}=[dot,fill=gray!40!white,,text depth=-0.2mm]
\tikzstyle{red dot}=[gray dot] % for backwards-compatibility

\tikzstyle{Z}=[white dot]
\tikzstyle{X}=[gray dot]
\tikzstyle{simple}=[]

\tikzstyle{black ddot}=[ddot,fill=black]
\tikzstyle{white ddot}=[ddot,fill=white]
\tikzstyle{gray ddot}=[ddot,fill=gray!40!white]

\tikzstyle{gray edge}=[gray!40!white]

\tikzstyle{small dot}=[inner sep=1pt,minimum width=0pt,minimum height=0pt,draw,shape=circle]

\tikzstyle{small black dot}=[small dot,fill=black]
\tikzstyle{small white dot}=[small dot,fill=white]
\tikzstyle{small gray dot}=[small dot,fill=gray!40!white]
\tikzstyle{special dot} = [small white dot]

\tikzstyle{mbqc dot}=[small black dot]
\tikzstyle{mbqc input dot}=[small white dot]
\tikzstyle{mbqc output dot}=[small gray dot]

\tikzstyle{causal dot}=[inner sep=0.4mm,minimum width=0pt,minimum height=0pt,draw=white,shape=circle,fill=gray!40!white]

\tikzstyle{phase dimensions}=[minimum size=5mm,font=\footnotesize,rectangle,rounded corners=2mm,inner sep=0.2mm,outer sep=-2mm,scale=0.8]
\tikzstyle{dphase dimensions}=[minimum size=5mm,font=\footnotesize,rectangle,rounded corners=2.5mm,inner sep=0.2mm,outer sep=-2mm]
\tikzstyle{white phase dot}=[dot,fill=white,phase dimensions]
\tikzstyle{white phase ddot}=[ddot,fill=white,dphase dimensions]

\tikzstyle{white rect ddot}=[draw=black,fill=white,doubled,minimum size=5mm,font=\footnotesize,rectangle,rounded corners=2.5mm,inner sep=0.2mm]
\tikzstyle{gray rect ddot}=[draw=black,fill=gray!40!white,doubled,minimum size=6mm,font=\footnotesize,rectangle,rounded corners=3mm]

\tikzstyle{gray phase dot}=[dot,fill=gray!40!white,phase dimensions]
\tikzstyle{gray phase ddot}=[ddot,fill=gray!40!white,dphase dimensions]
\tikzstyle{grey phase dot}=[gray phase dot]
\tikzstyle{grey phase ddot}=[gray phase ddot]

\tikzstyle{small phase dimensions}=[minimum size=4mm,font=\tiny,rectangle,rounded corners=2mm,inner sep=0.2mm,outer sep=-2mm]
\tikzstyle{small dphase dimensions}=[minimum size=4mm,font=\tiny,rectangle,rounded corners=2mm,inner sep=0.2mm,outer sep=-2mm]

\tikzstyle{small gray phase dot}=[dot,fill=gray!40!white,small phase dimensions]
\tikzstyle{small gray phase ddot}=[ddot,fill=gray!40!white,small dphase dimensions]

\tikzstyle{small map}=[draw,shape=rectangle,minimum height=4mm,minimum width=4mm,fill=white]

\tikzstyle{cnot}=[fill=white,shape=circle,inner sep=-1.4pt]

\tikzstyle{asym hadamard}=[fill=white,draw,shape=NEbox,inner sep=0.6mm,font=\footnotesize,minimum height=4mm]
\tikzstyle{asym hadamard conj}=[fill=white,draw,shape=NWbox,inner sep=0.6mm,font=\footnotesize,minimum height=4mm]
\tikzstyle{asym hadamard dag}=[fill=white,draw,shape=SEbox,inner sep=0.6mm,font=\footnotesize,minimum height=4mm]

\tikzstyle{hadamard}=[fill=white,draw,inner sep=0.6mm,font=\footnotesize,minimum height=4mm,minimum width=4mm]
\tikzstyle{small hadamard}=[fill=white,draw,inner sep=0.6mm,minimum height=1.5mm,minimum width=1.5mm]
\tikzstyle{small hadamard rotate}=[small hadamard,rotate=45]
\tikzstyle{dhadamard}=[hadamard,doubled]
\tikzstyle{small dhadamard}=[small hadamard,doubled]
\tikzstyle{small dhadamard rotate}=[small hadamard rotate,doubled]
\tikzstyle{antipode}=[white dot,inner sep=0.3mm,font=\footnotesize]

\tikzstyle{scalar}=[diamond,draw,inner sep=0.5pt,font=\small]
\tikzstyle{dscalar}=[diamond,doubled, draw,inner sep=0.5pt,font=\small]

\tikzstyle{small box}=[rectangle,inline text,fill=white,draw,minimum height=5mm,yshift=-0.5mm,minimum width=5mm,font=\small]
\tikzstyle{small gray box}=[small box,fill=gray!30]
\tikzstyle{medium box}=[rectangle,inline text,fill=white,draw,minimum height=5mm,yshift=-0.5mm,minimum width=10mm,font=\small]
\tikzstyle{square box}=[small box] % for backwards-compatibility
\tikzstyle{medium gray box}=[small box,fill=gray!30]
\tikzstyle{semilarge box}=[rectangle,inline text,fill=white,draw,minimum height=5mm,yshift=-0.5mm,minimum width=12.5mm,font=\small]
\tikzstyle{large box}=[rectangle,inline text,fill=white,draw,minimum height=5mm,yshift=-0.5mm,minimum width=15mm,font=\small]
\tikzstyle{large gray box}=[small box,fill=gray!30]

\tikzstyle{Bayes box}=[rectangle,fill=black,draw, minimum height=3mm, minimum width=3mm]

\tikzstyle{gray square point}=[small box,fill=gray!50]

\tikzstyle{dphase box white}=[dhadamard]
\tikzstyle{dphase box gray}=[dhadamard,fill=gray!50!white]
\tikzstyle{phase box white}=[hadamard]
\tikzstyle{phase box gray}=[hadamard,fill=gray!50!white]

\tikzstyle{point}=[regular polygon,regular polygon sides=3,draw,scale=0.75,inner sep=-0.5pt,minimum width=9mm,fill=white,regular polygon rotate=180]
\tikzstyle{copoint}=[regular polygon,regular polygon sides=3,draw,scale=0.75,inner sep=-0.5pt,minimum width=9mm,fill=white]
\tikzstyle{dpoint}=[point,doubled]
\tikzstyle{dcopoint}=[copoint,doubled]

\tikzstyle{wide copoint}=[fill=white,draw,shape=isosceles triangle,shape border rotate=90,isosceles triangle stretches=true,inner sep=0pt,minimum width=1.5cm,minimum height=6.12mm]
\tikzstyle{wide point}=[fill=white,draw,shape=isosceles triangle,shape border rotate=-90,isosceles triangle stretches=true,inner sep=0pt,minimum width=1.5cm,minimum height=6.12mm,yshift=-0.0mm]
\tikzstyle{wide point plus}=[fill=white,draw,shape=isosceles triangle,shape border rotate=-90,isosceles triangle stretches=true,inner sep=0pt,minimum width=1.74cm,minimum height=7mm,yshift=-0.0mm]

\tikzstyle{wide dpoint}=[fill=white,doubled,draw,shape=isosceles triangle,shape border rotate=-90,isosceles triangle stretches=true,inner sep=0pt,minimum width=1.5cm,minimum height=6.12mm,yshift=-0.0mm]

\tikzstyle{tinypoint}=[regular polygon,regular polygon sides=3,draw,scale=0.55,inner sep=-0.15pt,minimum width=6mm,fill=white,regular polygon rotate=180] 

\tikzstyle{white point}=[point]
\tikzstyle{white dpoint}=[dpoint]
\tikzstyle{green point}=[white point] % for backwards-compatibility
\tikzstyle{white copoint}=[copoint]
\tikzstyle{gray point}=[point,fill=gray!40!white]
\tikzstyle{gray dpoint}=[gray point,doubled]
\tikzstyle{red point}=[gray point] % for backwards-compatibility
\tikzstyle{gray copoint}=[copoint,fill=gray!40!white]
\tikzstyle{gray dcopoint}=[gray copoint,doubled]

\tikzstyle{white point guide}=[regular polygon,regular polygon sides=3,font=\scriptsize,draw,scale=0.65,inner sep=-0.5pt,minimum width=9mm,fill=white,regular polygon rotate=180]

\tikzstyle{black point}=[point,fill=black,font=\color{white}]
\tikzstyle{black copoint}=[copoint,fill=black,font=\color{white}]

\tikzstyle{tiny gray point}=[tinypoint,fill=gray!40!white]

\tikzstyle{diredge}=[->]
\tikzstyle{ddiredge}=[<->]
\tikzstyle{rdiredge}=[<-]
\tikzstyle{thickdiredge}=[->, very thick]
\tikzstyle{pointer edge}=[->,very thick,gray]
\tikzstyle{pointer edge part}=[very thick,gray]
\tikzstyle{dashed edge}=[dashed]
\tikzstyle{thick dashed edge}=[very thick,dashed]
\tikzstyle{thick gray dashed edge}=[thick dashed edge,gray!40]
\tikzstyle{thick map edge}=[very thick,|->]

\makeatletter
\newcommand{\boxshape}[3]{%
\pgfdeclareshape{#1}{
\inheritsavedanchors[from=rectangle] % this is nearly a rectangle
\inheritanchorborder[from=rectangle]
\inheritanchor[from=rectangle]{center}
\inheritanchor[from=rectangle]{north}
\inheritanchor[from=rectangle]{south}
\inheritanchor[from=rectangle]{west}
\inheritanchor[from=rectangle]{east}
\backgroundpath{% this is new
\southwest \pgf@xa=\pgf@x \pgf@ya=\pgf@y
\northeast \pgf@xb=\pgf@x \pgf@yb=\pgf@y

\@tempdima=#2
\@tempdimb=#3

\pgfpathmoveto{\pgfpoint{\pgf@xa - 5pt + \@tempdima}{\pgf@ya}}
\pgfpathlineto{\pgfpoint{\pgf@xa - 5pt - \@tempdima}{\pgf@yb}}
\pgfpathlineto{\pgfpoint{\pgf@xb + 5pt + \@tempdimb}{\pgf@yb}}
\pgfpathlineto{\pgfpoint{\pgf@xb + 5pt - \@tempdimb}{\pgf@ya}}
\pgfpathlineto{\pgfpoint{\pgf@xa - 5pt + \@tempdima}{\pgf@ya}}
\pgfpathclose
}
}}

\boxshape{NEbox}{0pt}{5pt}
\boxshape{SEbox}{0pt}{-5pt}
\boxshape{NWbox}{5pt}{0pt}
\boxshape{SWbox}{-5pt}{0pt}
\boxshape{EBox}{-3pt}{3pt}
\boxshape{WBox}{3pt}{-3pt}
\makeatother

\tikzstyle{cloud}=[shape=cloud,draw,minimum width=1.5cm,minimum height=1.5cm]

\tikzstyle{map}=[draw,shape=NEbox,inner sep=2pt,minimum height=6mm,fill=white]
\tikzstyle{dashedmap}=[draw,dashed,shape=NEbox,inner sep=2pt,minimum height=6mm,fill=white]
\tikzstyle{mapdag}=[draw,shape=SEbox,inner sep=2pt,minimum height=6mm,fill=white]
\tikzstyle{mapadj}=[draw,shape=SEbox,inner sep=2pt,minimum height=6mm,fill=white]
\tikzstyle{maptrans}=[draw,shape=SWbox,inner sep=2pt,minimum height=6mm,fill=white]
\tikzstyle{mapconj}=[draw,shape=NWbox,inner sep=2pt,minimum height=6mm,fill=white]

\tikzstyle{medium map}=[draw,shape=NEbox,inner sep=2pt,minimum height=6mm,fill=white,minimum width=7mm]
\tikzstyle{medium map dag}=[draw,shape=SEbox,inner sep=2pt,minimum height=6mm,fill=white,minimum width=7mm]
\tikzstyle{medium map adj}=[draw,shape=SEbox,inner sep=2pt,minimum height=6mm,fill=white,minimum width=7mm]
\tikzstyle{medium map trans}=[draw,shape=SWbox,inner sep=2pt,minimum height=6mm,fill=white,minimum width=7mm]
\tikzstyle{medium map conj}=[draw,shape=NWbox,inner sep=2pt,minimum height=6mm,fill=white,minimum width=7mm]
\tikzstyle{semilarge map}=[draw,shape=NEbox,inner sep=2pt,minimum height=6mm,fill=white,minimum width=9.5mm]
\tikzstyle{semilarge map trans}=[draw,shape=SWbox,inner sep=2pt,minimum height=6mm,fill=white,minimum width=9.5mm]
\tikzstyle{semilarge map adj}=[draw,shape=SEbox,inner sep=2pt,minimum height=6mm,fill=white,minimum width=9.5mm]
\tikzstyle{semilarge map dag}=[draw,shape=SEbox,inner sep=2pt,minimum height=6mm,fill=white,minimum width=9.5mm]
\tikzstyle{semilarge map conj}=[draw,shape=NWbox,inner sep=2pt,minimum height=6mm,fill=white,minimum width=9.5mm]
\tikzstyle{large map}=[draw,shape=NEbox,inner sep=2pt,minimum height=6mm,fill=white,minimum width=12mm]
\tikzstyle{large map conj}=[draw,shape=NWbox,inner sep=2pt,minimum height=6mm,fill=white,minimum width=12mm]
\tikzstyle{very large map}=[draw,shape=NEbox,inner sep=2pt,minimum height=6mm,fill=white,minimum width=17mm]

\tikzstyle{medium dmap}=[draw,doubled,shape=NEbox,inner sep=2pt,minimum height=6mm,fill=white,minimum width=7mm]
\tikzstyle{medium dmap dag}=[draw,doubled,shape=SEbox,inner sep=2pt,minimum height=6mm,fill=white,minimum width=7mm]
\tikzstyle{medium dmap adj}=[draw,doubled,shape=SEbox,inner sep=2pt,minimum height=6mm,fill=white,minimum width=7mm]
\tikzstyle{medium dmap trans}=[draw,doubled,shape=SWbox,inner sep=2pt,minimum height=6mm,fill=white,minimum width=7mm]
\tikzstyle{medium dmap conj}=[draw,doubled,shape=NWbox,inner sep=2pt,minimum height=6mm,fill=white,minimum width=7mm]
\tikzstyle{semilarge dmap}=[draw,doubled,shape=NEbox,inner sep=2pt,minimum height=6mm,fill=white,minimum width=9.5mm]
\tikzstyle{semilarge dmap trans}=[draw,doubled,shape=SWbox,inner sep=2pt,minimum height=6mm,fill=white,minimum width=9.5mm]
\tikzstyle{semilarge dmap adj}=[draw,doubled,shape=SEbox,inner sep=2pt,minimum height=6mm,fill=white,minimum width=9.5mm]
\tikzstyle{semilarge dmap dag}=[draw,doubled,shape=SEbox,inner sep=2pt,minimum height=6mm,fill=white,minimum width=9.5mm]
\tikzstyle{semilarge dmap conj}=[draw,doubled,shape=NWbox,inner sep=2pt,minimum height=6mm,fill=white,minimum width=9.5mm]
\tikzstyle{large dmap}=[draw,doubled,shape=NEbox,inner sep=2pt,minimum height=6mm,fill=white,minimum width=12mm]
\tikzstyle{large dmap conj}=[draw,doubled,shape=NWbox,inner sep=2pt,minimum height=6mm,fill=white,minimum width=12mm]
\tikzstyle{large dmap trans}=[draw,doubled,shape=SWbox,inner sep=2pt,minimum height=6mm,fill=white,minimum width=12mm]
\tikzstyle{large dmap adj}=[draw,doubled,shape=SEbox,inner sep=2pt,minimum height=6mm,fill=white,minimum width=12mm]
\tikzstyle{large dmap dag}=[draw,doubled,shape=SEbox,inner sep=2pt,minimum height=6mm,fill=white,minimum width=12mm]
\tikzstyle{very large dmap}=[draw,doubled,shape=NEbox,inner sep=2pt,minimum height=6mm,fill=white,minimum width=19.5mm]

\tikzstyle{muxbox}=[draw,shape=rectangle,minimum height=3mm,minimum width=3mm,fill=white]
\tikzstyle{dmuxbox}=[muxbox,doubled]

\tikzstyle{box}=[draw,shape=rectangle,inner sep=2pt,minimum height=6mm,minimum width=6mm,fill=white]
\tikzstyle{dbox}=[draw,doubled,shape=rectangle,inner sep=2pt,minimum height=6mm,minimum width=6mm,fill=white]
\tikzstyle{dmap}=[draw,doubled,shape=NEbox,inner sep=2pt,minimum height=6mm,fill=white]
\tikzstyle{dmapdag}=[draw,doubled,shape=SEbox,inner sep=2pt,minimum height=6mm,fill=white]
\tikzstyle{dmapadj}=[draw,doubled,shape=SEbox,inner sep=2pt,minimum height=6mm,fill=white]
\tikzstyle{dmaptrans}=[draw,doubled,shape=SWbox,inner sep=2pt,minimum height=6mm,fill=white]
\tikzstyle{dmapconj}=[draw,doubled,shape=NWbox,inner sep=2pt,minimum height=6mm,fill=white]

\tikzstyle{ddmap}=[draw,doubled,dashed,shape=NEbox,inner sep=2pt,minimum height=6mm,fill=white]
\tikzstyle{ddmapdag}=[draw,doubled,dashed,shape=SEbox,inner sep=2pt,minimum height=6mm,fill=white]
\tikzstyle{ddmapadj}=[draw,doubled,dashed,shape=SEbox,inner sep=2pt,minimum height=6mm,fill=white]
\tikzstyle{ddmaptrans}=[draw,doubled,dashed,shape=SWbox,inner sep=2pt,minimum height=6mm,fill=white]
\tikzstyle{ddmapconj}=[draw,doubled,dashed,shape=NWbox,inner sep=2pt,minimum height=6mm,fill=white]

\boxshape{sNEbox}{0pt}{3pt}
\boxshape{sSEbox}{0pt}{-3pt}
\boxshape{sNWbox}{3pt}{0pt}
\boxshape{sSWbox}{-3pt}{0pt}
\tikzstyle{smap}=[draw,shape=sNEbox,fill=white]
\tikzstyle{smapdag}=[draw,shape=sSEbox,fill=white]
\tikzstyle{smapadj}=[draw,shape=sSEbox,fill=white]
\tikzstyle{smaptrans}=[draw,shape=sSWbox,fill=white]
\tikzstyle{smapconj}=[draw,shape=sNWbox,fill=white]

\tikzstyle{dsmap}=[draw,dashed,shape=sNEbox,fill=white]
\tikzstyle{dsmapdag}=[draw,dashed,shape=sSEbox,fill=white]
\tikzstyle{dsmaptrans}=[draw,dashed,shape=sSWbox,fill=white]
\tikzstyle{dsmapconj}=[draw,dashed,shape=sNWbox,fill=white]

\boxshape{mNEbox}{0pt}{10pt}
\boxshape{mSEbox}{0pt}{-10pt}
\boxshape{mNWbox}{10pt}{0pt}
\boxshape{mSWbox}{-10pt}{0pt}
\tikzstyle{mmap}=[draw,shape=mNEbox]
\tikzstyle{mmapdag}=[draw,shape=mSEbox]
\tikzstyle{mmaptrans}=[draw,shape=mSWbox]
\tikzstyle{mmapconj}=[draw,shape=mNWbox]

\tikzstyle{mmapgray}=[draw,fill=gray!40!white,shape=mNEbox]
\tikzstyle{smapgray}=[draw,fill=gray!40!white,shape=sNEbox]

\makeatletter

\pgfdeclareshape{cornerpoint}{
\inheritsavedanchors[from=rectangle] % this is nearly a rectangle
\inheritanchorborder[from=rectangle]
\inheritanchor[from=rectangle]{center}
\inheritanchor[from=rectangle]{north}
\inheritanchor[from=rectangle]{south}
\inheritanchor[from=rectangle]{west}
\inheritanchor[from=rectangle]{east}
\backgroundpath{% this is new
\southwest \pgf@xa=\pgf@x \pgf@ya=\pgf@y
\northeast \pgf@xb=\pgf@x \pgf@yb=\pgf@y

\pgfmathsetmacro{\pgf@shorten@left}{\pgfkeysvalueof{/tikz/shorten left}}
\pgfmathsetmacro{\pgf@shorten@right}{\pgfkeysvalueof{/tikz/shorten right}}

\pgfpathmoveto{\pgfpoint{0.5 * (\pgf@xa + \pgf@xb)}{\pgf@ya - 5pt}}
\pgfpathlineto{\pgfpoint{\pgf@xa - 8pt + \pgf@shorten@left}{\pgf@yb - 1.5 * \pgf@shorten@left}}
\pgfpathlineto{\pgfpoint{\pgf@xa - 8pt + \pgf@shorten@left}{\pgf@yb}}
\pgfpathlineto{\pgfpoint{\pgf@xb + 8pt - \pgf@shorten@right}{\pgf@yb}}
\pgfpathlineto{\pgfpoint{\pgf@xb + 8pt - \pgf@shorten@right}{\pgf@yb - 1.5 * \pgf@shorten@right}}
\pgfpathclose
}
}

\pgfdeclareshape{cornercopoint}{
\inheritsavedanchors[from=rectangle] % this is nearly a rectangle
\inheritanchorborder[from=rectangle]
\inheritanchor[from=rectangle]{center}
\inheritanchor[from=rectangle]{north}
\inheritanchor[from=rectangle]{south}
\inheritanchor[from=rectangle]{west}
\inheritanchor[from=rectangle]{east}
\backgroundpath{% this is new
\southwest \pgf@xa=\pgf@x \pgf@ya=\pgf@y
\northeast \pgf@xb=\pgf@x \pgf@yb=\pgf@y

\pgfmathsetmacro{\pgf@shorten@left}{\pgfkeysvalueof{/tikz/shorten left}}
\pgfmathsetmacro{\pgf@shorten@right}{\pgfkeysvalueof{/tikz/shorten right}}

\pgfpathmoveto{\pgfpoint{0.5 * (\pgf@xa + \pgf@xb)}{\pgf@yb + 5pt}}
\pgfpathlineto{\pgfpoint{\pgf@xa - 8pt + \pgf@shorten@left}{\pgf@ya + 1.5 * \pgf@shorten@left}}
\pgfpathlineto{\pgfpoint{\pgf@xa - 8pt + \pgf@shorten@left}{\pgf@ya}}
\pgfpathlineto{\pgfpoint{\pgf@xb + 8pt - \pgf@shorten@right}{\pgf@ya}}
\pgfpathlineto{\pgfpoint{\pgf@xb + 8pt - \pgf@shorten@right}{\pgf@ya + 1.5 * \pgf@shorten@right}}
\pgfpathclose
}
}

\makeatother

\pgfkeyssetvalue{/tikz/shorten left}{0pt}
\pgfkeyssetvalue{/tikz/shorten right}{0pt}

\tikzstyle{kpoint common}=[draw,fill=white,inner sep=1pt,minimum height=4mm]
\tikzstyle{kpoint sc}=[shape=cornerpoint,kpoint common]
\tikzstyle{kpoint adjoint sc}=[shape=cornercopoint,kpoint common]
\tikzstyle{kpoint}=[shape=cornerpoint,shorten left=5pt,kpoint common]
\tikzstyle{kpoint adjoint}=[shape=cornercopoint,shorten left=5pt,kpoint common]
\tikzstyle{kpoint conjugate}=[shape=cornerpoint,shorten right=5pt,kpoint common]
\tikzstyle{kpoint transpose}=[shape=cornercopoint,shorten right=5pt,kpoint common]
\tikzstyle{kpoint symm}=[shape=cornerpoint,shorten left=5pt,shorten right=5pt,kpoint common]

\tikzstyle{black kpoint}=[shape=cornerpoint,shorten left=5pt,kpoint common,fill=black,font=\color{white}]
\tikzstyle{black kpoint adjoint}=[shape=cornercopoint,shorten left=5pt,kpoint common,fill=black,font=\color{white}]
\tikzstyle{black kpointadj}=[shape=cornercopoint,shorten left=5pt,kpoint common,fill=black,font=\color{white}]

\tikzstyle{black dkpoint}=[shape=cornerpoint,shorten left=5pt,kpoint common,fill=black, doubled,font=\color{white}]
\tikzstyle{black dkpoint adjoint}=[shape=cornercopoint,shorten left=5pt,kpoint common,fill=black, doubled,font=\color{white}]
\tikzstyle{black dkpointadj}=[shape=cornercopoint,shorten left=5pt,kpoint common,fill=black, doubled,font=\color{white}] 

\tikzstyle{kpointdag}=[kpoint adjoint]
\tikzstyle{kpointadj}=[kpoint adjoint]
\tikzstyle{kpointconj}=[kpoint conjugate]
\tikzstyle{kpointtrans}=[kpoint transpose]

\tikzstyle{big kpoint}=[kpoint, minimum width=1.2 cm, minimum height=8mm, inner sep=4pt, text depth=3mm]

\tikzstyle{wide kpoint}=[kpoint, minimum width=1 cm, inner sep=2pt]%, text depth=-0.7 mm]
\tikzstyle{wide kpointdag}=[kpointdag, minimum width=1 cm, inner sep=2pt]%, text depth=0.7 mm]
\tikzstyle{wide kpointconj}=[kpointconj, minimum width=1 cm, inner sep=2pt]%, text depth=-0.7 mm]
\tikzstyle{wide kpointtrans}=[kpointtrans, minimum width=1 cm, inner sep=2pt]%, text depth=0.7 mm]

\tikzstyle{gray kpoint}=[kpoint,fill=gray!50!white]
\tikzstyle{gray kpointdag}=[kpointdag,fill=gray!50!white]
\tikzstyle{gray kpointadj}=[kpointadj,fill=gray!50!white]
\tikzstyle{gray kpointconj}=[kpointconj,fill=gray!50!white]
\tikzstyle{gray kpointtrans}=[kpointtrans,fill=gray!50!white]

\tikzstyle{gray dkpoint}=[kpoint,fill=gray!50!white,doubled]
\tikzstyle{gray dkpointdag}=[kpointdag,fill=gray!50!white,doubled]
\tikzstyle{gray dkpointadj}=[kpointadj,fill=gray!50!white,doubled]
\tikzstyle{gray dkpointconj}=[kpointconj,fill=gray!50!white,doubled]
\tikzstyle{gray dkpointtrans}=[kpointtrans,fill=gray!50!white,doubled]

\tikzstyle{white label}=[draw,fill=white,rectangle,inner sep=0.7 mm]
\tikzstyle{gray label}=[draw,fill=gray!50!white,rectangle,inner sep=0.7 mm]
\tikzstyle{black label}=[draw,fill=black,rectangle,inner sep=0.7 mm]

\tikzstyle{dkpoint}=[kpoint,doubled]
\tikzstyle{wide dkpoint}=[wide kpoint,doubled]
\tikzstyle{dkpointdag}=[kpoint adjoint,doubled]
\tikzstyle{wide dkpointdag}=[wide kpointdag,doubled]
\tikzstyle{dkcopoint}=[kpoint adjoint,doubled]
\tikzstyle{dkpointadj}=[kpoint adjoint,doubled]
\tikzstyle{dkpointconj}=[kpoint conjugate,doubled]
\tikzstyle{dkpointtrans}=[kpoint transpose,doubled]

\tikzstyle{kscalar}=[kpoint common, shape=EBox, inner xsep=-1pt, inner ysep=3pt,font=\small]
\tikzstyle{kscalarconj}=[kpoint common, shape=WBox, inner xsep=-1pt, inner ysep=3pt,font=\small]

\tikzstyle{spekpoint}=[kpoint sc,minimum height=5mm,inner sep=3pt]
\tikzstyle{spekcopoint}=[kpoint adjoint sc,minimum height=5mm,inner sep=3pt]

\tikzstyle{dspekpoint}=[spekpoint,doubled]
\tikzstyle{dspekcopoint}=[spekcopoint,doubled]

 \tikzstyle{upground}=[circuit ee IEC,thick,ground,rotate=90,scale=1.0]
 \tikzstyle{downground}=[circuit ee IEC,thick,ground,rotate=-90,scale=1.0]
 \tikzstyle{bigground}=[regular polygon,regular polygon sides=3,draw=gray,scale=0.50,inner sep=-0.5pt,minimum width=10mm,fill=gray]
\tikzstyle{arrs}=[-latex,font=\small,auto]
\tikzstyle{arrow plain}=[arrs]
\tikzstyle{arrow dashed}=[dashed,arrs]
\tikzstyle{arrow bold}=[very thick,arrs]
\tikzstyle{arrow hide}=[draw=white!0,-]
\tikzstyle{arrow reverse}=[latex-]
\tikzstyle{cdnode}=[]

%% file: zx.tikzstyles
\tikzstyle{box}=[shape=rectangle, text height=1.5ex, text depth=0.25ex, yshift=0.5mm, fill=white, draw=black, minimum height=5mm, yshift=-0.5mm, minimum width=5mm, font={\small}]
\tikzstyle{Z dot}=[inner sep=0mm, minimum size=2mm, shape=circle, draw=black, fill={rgb,255: red,221; green,255; blue,221}]
\tikzstyle{Z phase dot}=[minimum size=5mm, font={\footnotesize\boldmath}, shape=rectangle, rounded corners=2mm, inner sep=0.2mm, outer sep=-2mm, scale=0.8, tikzit shape=circle, draw=black, fill={rgb,255: red,221; green,255; blue,221}, tikzit draw=blue]
\tikzstyle{X dot}=[Z dot, shape=circle, draw=black, fill={rgb,255: red,255; green,136; blue,136}]
\tikzstyle{X phase dot}=[Z phase dot, tikzit shape=circle, tikzit draw=blue, fill={rgb,255: red,255; green,136; blue,136}, font={\footnotesize\boldmath}]
\tikzstyle{hadamard}=[fill=yellow, draw=black, shape=rectangle, inner sep=0.6mm, minimum height=1.5mm, minimum width=1.5mm]
\tikzstyle{vertex}=[inner sep=0mm, minimum size=1mm, shape=circle, draw=black, fill=black]
\tikzstyle{vertex set}=[inner sep=0mm, minimum size=1mm, shape=circle, draw=black, fill=white, font={\footnotesize\boldmath}]
\tikzstyle{ground}=[ground]

\tikzstyle{hadamard edge}=[-, dashed, dash pattern=on 2pt off 0.5pt, thick, draw={rgb,255: red,68; green,136; blue,255}]
\tikzstyle{brace edge}=[-, tikzit draw=blue, decorate, decoration={brace,amplitude=1mm,raise=-1mm}]
\tikzstyle{diredge}=[->]

\tikzstyle{border edge}=[-, dashed, dash pattern=on 2pt off 0.5pt, thick, draw={rgb,255: red,255; green,13; blue,20}]

\tikzstyle{bit clone}=[small black dot]

\tikzstyle{qubit edge}=[-]
\tikzstyle{bit edge}=[-, double distance=1.5pt]

\tikzstyle{bang box}=[shape=rectangle, text height=1.5ex, text depth=0.25ex, yshift=0.5mm, draw=black, minimum height=10mm, yshift=-0.5mm, minimum width=10mm, font={\small},
dashed]

%% file: defs.tex
\newcommand{\bigo}{\ensuremath{\mathcal{O}}}
\newcommand{\C}{\ensuremath{\mathbb{C}}}
\newcommand{\Dn}[1]{\ensuremath{\mathcal{D}_{#1}}} % Set of n-qubit density matrices

\newcommand{\R}{\ensuremath{\mathbb{R}}}

\newcommand{\ketbra}[2]{\ket{#1}\!\!\bra{#2}}

\newcommand{\ketR}{\ket{\circlearrowright}}
\newcommand{\ketL}{\ket{\circlearrowleft}}

\usetikzlibrary{circuits.ee.IEC}
\usetikzlibrary{shapes.gates.logic.US,shapes.gates.logic.IEC}

\newcommand{\ground}{%
	\begin{tikzpicture}[circuit ee IEC,yscale=1.0,xscale=1.0]
	\draw[solid,arrows=-] (0,1ex) to (0,0) node[anchor=center,ground,rotate=-90,xshift=.66ex] {};
	\end{tikzpicture}
}

\newcommand{\zxGND}{\ensuremath{\text{ZX}_{\ground}}}
\newcommand{\zxGNDsafe}{\texorpdfstring{\zxGND}{ZX-ground}}
\newcommand{\interpretZX}[1]{\left\llbracket{#1}\right\rrbracket}

\newcommand{\toZxGND}[1]{T({#1})}

\newcommand{\outputs}{\mathsf{O}}
\newcommand{\inputs}{\mathsf{I}}

\newcommand{\AND}{\ensuremath{\text{AND}}}
\newcommand{\OR}{\ensuremath{\text{OR}}}
\newcommand{\NOT}{\ensuremath{\text{NOT}}}
\newcommand{\XOR}{\ensuremath{\text{XOR}}}
\newcommand{\CNOT}{\ensuremath{\text{CNOT}}}
\renewcommand{\H}{\ensuremath{\text{H}}}
\newcommand{\Z}[1][]{\ensuremath{\text{Z}_{#1}}}
\newcommand{\Rx}[1][]{\ensuremath{\text{X}_{#1}}}
\newcommand{\CZ}{\ensuremath{\text{CZ}}}

\newcommand{\odd}{\mathsf{Odd}}

\newcommand{\frontier}{\mathsf{F}}
\newcommand{\fQubit}{\mathsf{Q}}

\newcommand{\labels}{\ensuremath{\mathcal{L}}}
\newcommand{\Q}{\textsf{Q}}
\newcommand{\X}{\textsf{X}}
\newcommand{\Y}{\textsf{Y}}
\newcommand{\tJoin}{\star}
\newcommand{\rot}[1][\alpha]{\text{rot}_{#1}}
\DeclareMathOperator{\hadamLbl}{\H}
\newcommand{\interpretLbl}[1]{#1}
\newcommand{\ClassHadamRule}{\ensuremath{(\bm{ch})}}
\newcommand{\ClassZRule}{\ensuremath{(\bm{cz})}}

\definecolor{orange}{RGB}{255,165,0}
\spnewtheorem*{lemma*}{Lemma}{\bfseries}{\rmfamily}

%% file: figures/circuit/hadamard.tikz
\begin{tikzpicture}
	\begin{pgfonlayer}{nodelayer}
		\node [style=box] (0) at (1.5, 0) {$H$};
		\node [style=none] (1) at (3, 0) {};
		\node [style=none] (2) at (0, 0) {};
	\end{pgfonlayer}
	\begin{pgfonlayer}{edgelayer}
		\draw [in=180, out=0] (0) to (1.center);
		\draw [in=180, out=0] (2.center) to (0);
	\end{pgfonlayer}
\end{tikzpicture}

%% file: figures/circuit/z.tikz
\begin{tikzpicture}
	\begin{pgfonlayer}{nodelayer}
		\node [style=box] (0) at (1.5, 0) {$Z_\alpha$};
		\node [style=none] (1) at (3, 0) {};
		\node [style=none] (2) at (0, 0) {};
	\end{pgfonlayer}
	\begin{pgfonlayer}{edgelayer}
		\draw [in=180, out=0] (0) to (1.center);
		\draw [in=180, out=0] (2.center) to (0);
	\end{pgfonlayer}
\end{tikzpicture}

%% file: figures/circuit/x.tikz
\begin{tikzpicture}
	\begin{pgfonlayer}{nodelayer}
		\node [style=box] (0) at (1.5, 0) {$X_\alpha$};
		\node [style=none] (1) at (3, 0) {};
		\node [style=none] (2) at (0, 0) {};
	\end{pgfonlayer}
	\begin{pgfonlayer}{edgelayer}
		\draw [in=180, out=0] (0) to (1.center);
		\draw [in=180, out=0] (2.center) to (0);
	\end{pgfonlayer}
\end{tikzpicture}

%% file: figures/circuit/init.tikz
\begin{tikzpicture}
	\begin{pgfonlayer}{nodelayer}
		\node [style=none, xscale=2] (0) at (0, 0) {\textbf{\textbar}};
    \node at ([shift={(0:-.5)}]0.center) {$0$};
		\node [style=none] (1) at (1, 0) {};
	\end{pgfonlayer}
	\begin{pgfonlayer}{edgelayer}
		\draw [in=180, out=0] (0.center) to (1.center);
	\end{pgfonlayer}
\end{tikzpicture}

%% file: figures/circuit/termination.tikz
\begin{tikzpicture}
	\begin{pgfonlayer}{nodelayer}
		\node [style=none] (0) at (0, 0) {};
		\node [style=none, xscale=2] (1) at (1, 0) {\textbf{\textbar}};
    \node at ([shift={(0:.5)}]1.center) {$0$};
	\end{pgfonlayer}
	\begin{pgfonlayer}{edgelayer}
		\draw [in=180, out=0] (0.center) to (1.center);
	\end{pgfonlayer}
\end{tikzpicture}

%% file: figures/circuit/not.tikz
\begin{tikzpicture}
	\begin{pgfonlayer}{nodelayer}
		\node [not gate US, draw] (notGate) at (1, 0) {};
		\node [style=none] (1) at (2.5, 0) {};
		\node [style=none] (2) at (0, 0) {};
	\end{pgfonlayer}
	\begin{pgfonlayer}{edgelayer}
		\draw [bit edge, in=180, out=0] (notGate) to (1.center);
		\draw [bit edge, in=180, out=0] (2.center) to (notGate);
	\end{pgfonlayer}
\end{tikzpicture}

%% file: figures/circuit/prepare.tikz
\begin{tikzpicture}
	\begin{pgfonlayer}{nodelayer}
		\node [style=none] (2) at (-1, 0) {};
		\node [style=none, xscale=2] (0) at (0, 0) {\textbf{\textbar}};
		\node [style=none] (1) at (1, 0) {};
	\end{pgfonlayer}
	\begin{pgfonlayer}{edgelayer}
		\draw [bit edge] (2.center) to (0.center);
		\draw [qubit edge] (1.center) to (0.center);
	\end{pgfonlayer}
\end{tikzpicture}

%% file: figures/zxGnd/elem/hadamard.tikz
\begin{tikzpicture}
	\begin{pgfonlayer}{nodelayer}
		\node [style=hadamard] (0) at (0, 0) {};
		\node [style=none] (1) at (1, 0) {};
		\node [style=none] (2) at (-1, 0) {};
	\end{pgfonlayer}
	\begin{pgfonlayer}{edgelayer}
		\draw [in=180, out=0, looseness=0.75] (0) to (1.center);
		\draw [in=180, out=0, looseness=0.75] (2.center) to (0);
	\end{pgfonlayer}
\end{tikzpicture}

%% file: figures/zxGnd/rule/numberRule-a.tikz
\begin{tikzpicture}
	\begin{pgfonlayer}{nodelayer}
		\node [style=ground,rotate=-90] (0) at (0, 0) {\ground};
		\node [style=Z dot] (1) at (1, 0) {};
	\end{pgfonlayer}
	\begin{pgfonlayer}{edgelayer}
		\draw [out=0,in=180,looseness=1.50] (0.center) to (1.center);
	\end{pgfonlayer}
\end{tikzpicture}

%% file: figures/unit-diagram.tikz
\begin{tikzpicture}
	\begin{pgfonlayer}{nodelayer}
		\node [style=none] (0) at (-.5, -.5) {};
		\node [style=none] (1) at (-.5, 0.5) {};
		\node [style=none] (2) at (0.5, 0.5) {};
		\node [style=none] (3) at (0.5, -.5) {};
	\end{pgfonlayer}
	\begin{pgfonlayer}{edgelayer}
		\draw [style=dashed] (0.center) to (3.center);
		\draw [style=dashed] (1.center) to (0.center);
		\draw [style=dashed] (2.center) to (1.center);
		\draw [style=dashed] (3.center) to (2.center);
  \end{pgfonlayer}
\end{tikzpicture}

%% file: figures/zxGnd/rule/HRule-a.tikz
\begin{tikzpicture}
	\begin{pgfonlayer}{nodelayer}
		\node [style=ground, rotate=-90] (1) at (0, 0) {\ground};
		\node [style=hadamard] (2) at (1, 0) {};
		\node [style=none] (3) at (2, 0) {};
	\end{pgfonlayer}
	\begin{pgfonlayer}{edgelayer}
		\draw [out=0,in=180,looseness=1.50] (1.center) to (2);
		\draw [out=0,in=180,looseness=1.50] (2) to (3.center);
	\end{pgfonlayer}
\end{tikzpicture}

%% file: figures/zxGnd/rule/HRule-b.tikz
\begin{tikzpicture}
	\begin{pgfonlayer}{nodelayer}
		\node [style=ground, rotate=-90] (1) at (0, 0) {\ground};
		\node [style=none] (2) at (1, 0) {};
	\end{pgfonlayer}
	\begin{pgfonlayer}{edgelayer}
		\draw [out=0,in=180,looseness=1.50] (1.center) to (2.center);
	\end{pgfonlayer}
\end{tikzpicture}

%% file: figures/zxGnd/rule/MergeRule-b.tikz
\begin{tikzpicture}
	\begin{pgfonlayer}{nodelayer}
		\node [style=ground, rotate=-90] (1) at (0, 0) {\ground};
		\node [style=none] (2) at (1, 0) {};
	\end{pgfonlayer}
	\begin{pgfonlayer}{edgelayer}
		\draw [out=0,in=180,looseness=1.50] (1.center) to (2.center);
	\end{pgfonlayer}
\end{tikzpicture}

%% file: figures/zxGnd/rule/DoubleRule-b.tikz
\begin{tikzpicture}
	\begin{pgfonlayer}{nodelayer}
		\node [style=ground, rotate=-90] (gnd) at (0, 0) {\ground};
		\node [style=Z dot] (2) at (1.5, 0) {};
		\node [style=none] (3) at (2.5, 0) {};
	\end{pgfonlayer}
	\begin{pgfonlayer}{edgelayer}
		\draw [out=0,in=180,looseness=1.50] (gnd.center) to (2);
		\draw [out=0,in=180,looseness=1.50] (2) to (3.center);
	\end{pgfonlayer}
\end{tikzpicture}

%% file: figures/zxGnd/elem/hWire.tikz
\begin{tikzpicture}
	\begin{pgfonlayer}{nodelayer}
		\node [style=Z dot] (0) at (-1, 0) {};
		\node [style=Z dot] (5) at (1, 0) {};
		\node [style=none] (1) at (1.75, 0) {};
		\node [style=none] (2) at (-1.75, 0) {};
	\end{pgfonlayer}
	\begin{pgfonlayer}{edgelayer}
		\draw (2.center) to (0);
		\draw (5) to (1.center);
		\draw [style=hadamard edge] (0) to (5);
	\end{pgfonlayer}
\end{tikzpicture}

%% file: figures/zxGnd/bangBox-0.tikz
\begin{tikzpicture}
	\begin{pgfonlayer}{nodelayer}
		\node [style=Z dot] (0) at (-0.75, 0) {};
		\node [style=none] (53) at (-1.75, 0) {};
	\end{pgfonlayer}
	\begin{pgfonlayer}{edgelayer}
		\draw (53.center) to (0);
	\end{pgfonlayer}
\end{tikzpicture}

%% file: figures/extraction-v2/row-sum-0.tikz
\begin{tikzpicture}
	\begin{pgfonlayer}{nodelayer}
		\node [style=none] (0) at (-5.25, 0.75) {};
		\node [style=Z dot] (2) at (-4.5, 0.75) {};
		\node [style=Z dot] (5) at (1, 1) {};
		\node [style=Z dot] (8) at (-4.5, -1.25) {};
		\node [style=none] (10) at (-5.25, -1.25) {};
		\node [style=none] (21) at (-3.75, 1.25) {};
		\node [style=none] (22) at (-2.25, 1.25) {};
		\node [style=none] (23) at (-3.75, -1.75) {};
		\node [style=none] (24) at (-2.25, -1.75) {};
		\node [style=none] (25) at (-3, -0.25) {$D$};
		\node [style=none] (26) at (-2.25, 0) {};
		\node [style=none] (27) at (-2.25, 1) {};
		\node [style=none] (28) at (-4.5, 0) {$\vdots$};
		\node [style=none] (29) at (-3.75, -1.25) {};
		\node [style=none] (30) at (-3.75, 0.75) {};
		\node [style=Z dot] (31) at (1, 0) {};
		\node [style=Z dot] (32) at (2, 1) {};
		\node [style=Z dot] (33) at (2, 0) {};
		\node [style=Z dot] (34) at (-1.5, 1) {};
		\node [style=Z dot] (35) at (-1.5, 0) {};
		\node [style=Z dot] (36) at (-1.5, -1.5) {};
		\node [style=none] (37) at (-2.25, -1.5) {};
		\node [style=Z dot] (38) at (1, -1.5) {};
		\node [style=Z dot] (39) at (2, -1.5) {};
		\node [style=none] (40) at (-1.5, -0.5) {$\vdots$};
		\node [style=none] (41) at (1.5, -0.5) {$\vdots$};
		\node [style=none] (42) at (-0.25, 0) {$\vdots$};
		\node [style=none] (43) at (-0.75, 1) {};
		\node [style=none] (44) at (0.25, 1) {};
		\node [style=none] (45) at (-0.75, 0.25) {};
		\node [style=none] (46) at (0.25, -0.75) {};
		\node [style=none] (47) at (-0.75, -0.5) {};
		\node [style=none] (48) at (0.25, -1) {};
		\node [style=none] (49) at (0.25, -0.5) {};
		\node [style=none] (50) at (-0.75, -1) {};
		\node [style=none] (51) at (-1.5, -2) {};
		\node [style=none] (52) at (-0.25, -2.75) {$M$};
		\node [style=none] (53) at (1, -2) {};
	\end{pgfonlayer}
	\begin{pgfonlayer}{edgelayer}
		\draw (0.center) to (2);
		\draw (10.center) to (8);
		\draw (21.center) to (22.center);
		\draw (22.center) to (24.center);
		\draw (24.center) to (23.center);
		\draw (23.center) to (21.center);
		\draw [style=hadamard edge] (29.center) to (8);
		\draw [style=hadamard edge] (30.center) to (2);
		\draw (31) to (33);
		\draw (5) to (32);
		\draw (34) to (27.center);
		\draw (35) to (26.center);
		\draw (36) to (37.center);
		\draw (38) to (39);
		\draw [style=hadamard edge] (34) to (43.center);
		\draw [style=hadamard edge] (34) to (45.center);
		\draw [style=hadamard edge] (35) to (47.center);
		\draw [style=hadamard edge] (48.center) to (38);
		\draw [style=hadamard edge] (46.center) to (38);
		\draw [style=hadamard edge] (44.center) to (5);
		\draw [style=hadamard edge] (49.center) to (31);
		\draw [style=hadamard edge] (36) to (50.center);
		\draw [decorate, decoration={brace,amplitude=5pt,mirror}] (51.center) to (53.center);
	\end{pgfonlayer}
\end{tikzpicture}

%% file: figures/extraction-v2/row-sum-1.tikz
\begin{tikzpicture}
	\begin{pgfonlayer}{nodelayer}
		\node [style=none] (0) at (-5.25, 0.75) {};
		\node [style=Z dot] (2) at (-4.5, 0.75) {};
		\node [style=Z dot] (5) at (1, 1) {};
		\node [style=Z dot] (8) at (-4.5, -1.25) {};
		\node [style=none] (10) at (-5.25, -1.25) {};
		\node [style=none] (21) at (-3.75, 1.25) {};
		\node [style=none] (22) at (-2.25, 1.25) {};
		\node [style=none] (23) at (-3.75, -1.75) {};
		\node [style=none] (24) at (-2.25, -1.75) {};
		\node [style=none] (25) at (-3, -0.25) {$D$};
		\node [style=none] (26) at (-2.25, 0) {};
		\node [style=none] (27) at (-2.25, 1) {};
		\node [style=none] (28) at (-4.5, 0) {$\vdots$};
		\node [style=none] (29) at (-3.75, -1.25) {};
		\node [style=none] (30) at (-3.75, 0.75) {};
		\node [style=Z dot] (31) at (1, 0) {};
		\node [style=Z dot] (32) at (2.5, 1) {};
		\node [style=Z dot] (33) at (2.5, 0) {};
		\node [style=Z dot] (34) at (-1.5, 1) {};
		\node [style=Z dot] (35) at (-1.5, 0) {};
		\node [style=Z dot] (36) at (-1.5, -1.5) {};
		\node [style=none] (37) at (-2.25, -1.5) {};
		\node [style=Z dot] (38) at (1, -1.5) {};
		\node [style=Z dot] (39) at (2.5, -1.5) {};
		\node [style=none] (40) at (-1.5, -0.5) {$\vdots$};
		\node [style=none] (41) at (1.75, -0.5) {$\vdots$};
		\node [style=none] (42) at (-0.25, 0) {$\vdots$};
		\node [style=none] (43) at (-0.75, 1) {};
		\node [style=none] (44) at (0.25, 1) {};
		\node [style=none] (45) at (-0.75, 0.5) {};
		\node [style=none] (46) at (0.25, -0.75) {};
		\node [style=none] (47) at (-0.75, -0.5) {};
		\node [style=none] (48) at (0.25, -1) {};
		\node [style=none] (49) at (0.25, -0.5) {};
		\node [style=none] (50) at (-0.75, -1) {};
		\node [style=none] (51) at (-1.5, -2) {};
		\node [style=none] (52) at (-0.25, -2.75) {$M'$};
		\node [style=none] (53) at (1, -2) {};
		\node [style=X dot] (54) at (1.75, 1) {};
		\node [style=Z dot] (55) at (1.75, 0) {};
		\node [style=none] (56) at (0.25, 0.25) {};
		\node [style=none] (57) at (-0.75, -0.75) {};
	\end{pgfonlayer}
	\begin{pgfonlayer}{edgelayer}
		\draw (0.center) to (2);
		\draw (10.center) to (8);
		\draw (21.center) to (22.center);
		\draw (22.center) to (24.center);
		\draw (24.center) to (23.center);
		\draw (23.center) to (21.center);
		\draw [style=hadamard edge] (29.center) to (8);
		\draw [style=hadamard edge] (30.center) to (2);
		\draw (34) to (27.center);
		\draw (35) to (26.center);
		\draw (36) to (37.center);
		\draw (38) to (39);
		\draw [style=hadamard edge] (34) to (43.center);
		\draw [style=hadamard edge] (35) to (47.center);
		\draw [style=hadamard edge] (48.center) to (38);
		\draw [style=hadamard edge] (46.center) to (38);
		\draw [style=hadamard edge] (44.center) to (5);
		\draw [style=hadamard edge] (49.center) to (31);
		\draw [style=hadamard edge] (36) to (50.center);
		\draw [decorate, decoration={brace,amplitude=5pt,mirror}] (51.center) to (53.center);
		\draw (31) to (55);
		\draw (55) to (33);
		\draw (32) to (54);
		\draw (54) to (5);
		\draw (54) to (55);
		\draw [style=hadamard edge] (34) to (45.center);
		\draw [style=hadamard edge] (36) to (57.center);
		\draw [style=hadamard edge] (56.center) to (5);
	\end{pgfonlayer}
\end{tikzpicture}

%% file: figures/classical/id-za.tikz
\begin{tikzpicture}
	\begin{pgfonlayer}{nodelayer}
		\node [style=none] (0) at (0, 0) {};
		\node [style=none] (5) at (-1.5, 0) {};
		\node [style=none] (6) at (-1.3, 0.35) {\tiny$\Z$};
		\node [style=none] (7) at (-0.2, 0.35) {\tiny$A$};
	\end{pgfonlayer}
	\begin{pgfonlayer}{edgelayer}
		\draw [-] (0.center) to (5.center);
	\end{pgfonlayer}
\end{tikzpicture}

%% file: figures/classical/id-az.tikz
\begin{tikzpicture}
	\begin{pgfonlayer}{nodelayer}
		\node [style=none] (0) at (0, 0) {};
		\node [style=none] (5) at (-1.5, 0) {};
		\node [style=none] (6) at (-1.3, 0.35) {\tiny$A$};
		\node [style=none] (7) at (-0.2, 0.35) {\tiny$\Z$};
	\end{pgfonlayer}
	\begin{pgfonlayer}{edgelayer}
		\draw [-] (0.center) to (5.center);
	\end{pgfonlayer}
\end{tikzpicture}

%% file: appendix/semantics.tex
\section{Semantics of the ZX-ground calculus}%
\label{sec:semantics}

Carette et al.~\cite{CJPV19completenessMix} define an interpretation of \zxGND\
diagrams using a CPM construction (cf. Reference~\cite{Selinger07DCClosCat}). We
describe it here without the categorical language, as an interpretation of
diagrams into density matrices and completely positive maps modulo scalars.

Let $\Dn{n} \subseteq \C^{2^n \times 2^n}$ be the set of n-qubit density
matrices.  There exist a functor $\interpretZX{\cdot} : \zxGND \to
\mathbf{CPM(Qubit)}$ which associates to any diagram $D : n \to m$ a completely
positive map $\interpretZX{D} : \Dn{n} \to \Dn{m}$, inductively defined as
follows.

\begin{spreadlines}{1em}
\begin{gather*}%
  \interpretZX{D_1 \otimes D_2}
  \;:=\;
  \interpretZX{D_1} \otimes \interpretZX{D_2}
  \quad
  \interpretZX{D_2 \circ D_1}
  \;:=\;
  \interpretZX{D_2} \circ \interpretZX{D_1}
  \quad
  \interpretZX{\input{unit-diagram.tikz}}
  \;:=\;
  \begin{pmatrix} 1 \end{pmatrix}
  \qquad
  \interpretZX{\input{zxGnd/elem/spiderZsingle.tikz}}
  \;:=\;
    \begin{pmatrix}1+e^{2i\alpha}\end{pmatrix}
  \\
  \interpretZX{\input{zxGnd/elem/wire.tikz}}
  \;:=\;
    id
  \qquad
  \interpretZX{\input{zxGnd/elem/swap.tikz}}
  \;:=\; \rho \mapsto U~\rho~U^\dagger
  {\scriptstyle
    \text{\ where\ }
    U = \ketbra{00}{00} + \ketbra{01}{10} + \ketbra{10}{01} + \ketbra{11}{11}
  }
  \\
  \interpretZX{\input{zxGnd/elem/hadamard.tikz}}
  \;:=\; \rho \mapsto H~\rho~H^\dagger \text{\ where\ }
  H = \frac{1}{\sqrt 2} \begin{pmatrix}1 & 1 \\ 1 & -1\end{pmatrix}
  \\
  \interpretZX{\input{zxGnd/elem/ground.tikz}}
  \;:=\; \ketbra{0}{0} + \ketbra{1}{1}
  \qquad
  \interpretZX{\input{zxGnd/elem/groundEnd.tikz}}
  \;:=\; \rho \mapsto \bra{0} \rho \ket{0} + \bra{1} \rho \ket{1}
  \\
  \interpretZX{\input{zxGnd/elem/spiderZ.tikz}}
  \;:=\; \rho \mapsto A~\rho~A^\dagger
  {\scriptstyle
    \text{\ if $n+m>0$, where\ }
    A = \ketbra{0^m}{0^n} + e^{i\alpha} \ketbra{1^m}{1^n}
  }
  \\
  \interpretZX{\input{zxGnd/elem/spiderX.tikz}}
  \;:=\;
  \interpretZX{\input{zxGnd/elem/hadamard.tikz}}^{\otimes m} \circ
  \interpretZX{\input{zxGnd/elem/spiderZ.tikz}} \circ
  \interpretZX{\input{zxGnd/elem/hadamard.tikz}}^{\otimes n}
\end{gather*}
\end{spreadlines}

From the compositional rules, we can see that a ground attached to a green
spider corresponds to a measurement over the computational basis.
\[
  \interpretZX{\input{zxGnd/elem/spiderZgnd.tikz}}
  \;:=\; \rho \mapsto \ketbra{0^m}{0^n}~\rho~\ketbra{0^n}{0^m} + \ketbra{1^m}{1^n}~\rho~\ketbra{1^n}{1^m}
\]

It follows from rule \HadamardRule\ that the red \ground-spider corresponds to a
measurement over the diagonal basis:
\[
  \interpretZX{\input{zxGnd/elem/spiderXgnd.tikz}}
  \;:=\; \rho \mapsto \ketbra{+^m}{+^n}~\rho~\ketbra{+^n}{+^m} + \ketbra{-^m}{-^n}~\rho~\ketbra{-^n}{-^m}
\]

Notice that, in accordance to rule \GndDiscardRule, the phase of a
\ground-spider is irrelevant.

%% file: figures/zxGnd/elem/spiderZsingle.tikz
\begin{tikzpicture}
	\begin{pgfonlayer}{nodelayer}
		\node [style=Z phase dot] (0) at (0, 0) {$\alpha$};
	\end{pgfonlayer}
\end{tikzpicture}

%% file: figures/zxGnd/elem/wire.tikz
\begin{tikzpicture}
	\begin{pgfonlayer}{nodelayer}
		\node [style=none] (1) at (2, 0) {};
		\node [style=none] (2) at (0, 0) {};
	\end{pgfonlayer}
	\begin{pgfonlayer}{edgelayer}
		\draw [in=180, out=0] (2.center) to (1.center);
	\end{pgfonlayer}
\end{tikzpicture}

%% file: figures/zxGnd/elem/ground.tikz
\begin{tikzpicture}
	\begin{pgfonlayer}{nodelayer}
		\node [style=ground, rotate=-90] (0) at (0, 0) {$\ground$};
		\node [style=none] (1) at (1, 0) {};
	\end{pgfonlayer}
	\begin{pgfonlayer}{edgelayer}
		\draw [in=180, out=0, looseness=0.75] (0.center) to (1);
	\end{pgfonlayer}
\end{tikzpicture}

%% file: figures/zxGnd/elem/groundEnd.tikz
\begin{tikzpicture}
	\begin{pgfonlayer}{nodelayer}
		\node [style=none] (0) at (0, 0) {};
		\node [style=ground, rotate=90] (1) at (1, 0) {$\ground$};
	\end{pgfonlayer}
	\begin{pgfonlayer}{edgelayer}
		\draw [in=180, out=0, looseness=0.75] (0) to (1.center);
	\end{pgfonlayer}
\end{tikzpicture}

%% file: appendix/derivations.tex
\section{Extended proofs}%
\label{sec:appendix-proofs}

\subsection{Proofs of section~\ref{sec:graph-like}}%

\begin{lemma*}[\ref{lem:weak-to-strict-graphlike}]%
  \textit{%
  There exists an algorithm to transform an arbitrary \zxGND diagram into an
  equivalent strictly graph-like diagram.
  }
\end{lemma*}

\begin{proof}
  The strict graph-like condition limits to 1 the number of input, output, and
  ground generators connected to each spider, in addition to the weakly
  graph-like restrictions. By Proposition~\ref{lem:eq-to-weak-graphLike}, the
  diagram can be rewritten into an equivalent weakly graph-like diagram $D$.  We
  describe an algorithm which modifies $D$ to comply with the additional
  restriction.

  For each spider $v$ in $D$ connected to at least two inputs, outputs or
  \ground generators, add two Z-spiders and Hadamard wires to each connected
  input and output as follows.
  \begin{equation}%
    \label{eq:toGL}
    \input{zxGnd/toGL.tikz}
  \end{equation}
  Notice that $v$ may have at most one connected input.
  The introduced spiders correspond to identity operations, and therefore the
  procedure generates a strictly graph-like diagram equivalent to $D$.
\end{proof}

\subsection{Proofs of section~\ref{sec:translation}}%

\begin{lemma*}[\ref{lem:circ-to-graphLike}]%
  \textit{%
    The \zxGND-diagram resulting from the translation $\toZxGND{\cdot}$ is
    weakly graph-like.
  }
\end{lemma*}

\begin{proof}%
  Notice that all the translation rules aside from the serial composition
  generate weakly graph-like diagrams trivially.

  For the serial composition, notice that both $\toZxGND{C}$ and $\toZxGND{C'}$
  are weakly  graph-like by the induction hypothesis and therefore all spiders
  in the resulting diagram are Z-spiders, all inputs and outputs are connected
  to Z-spiders, no two inputs are connected to the same spider, and all spiders
  connected to \ground-generators have phase 0.  The internal edges added by the
  composition will therefore connect two green spiders, which will be merged by
  the fusion rule application.

  The fusion step may create spiders connected to two ground generators, one of
  which is removed by the application of rule \GndDoubleRule.
  It may also generate parallel Hadamard wires, which are removed by the
  application of rule \ParHRule.

  Therefore, the translation generates weakly graph-like \zxGND-diagrams.
\end{proof}

For the proof of Lemma~\ref{lem:weakGL-gFlow}, we use a weaker version of the
focused gFlow invariant called
\textit{causal flow}~\citep{mhallaWhichGraphStates2014}.

\begin{definition}
  Given an open graph $G$, a \textit{causal flow} $(f,\prec)$ on $G$ consists of a function
  $f: \overline{T} \to \overline{S}$ and a partial order $\prec$ on the set $V$ satisfying the
  following properties:
  \begin{enumerate}
    \item $f(v) \sim v$
    \item $v \prec f(v)$
    \item If $u \sim f(v)$ and $u \neq v$ then $v \prec u$
  \end{enumerate}
\end{definition}
where $u \sim v$ if $u$ is connected to $v$ in the graph, and $\overline{A} =
V(G) / A$.

Duncan et al.~\cite[Theorem B.3]{DKPW2019qcircSimpl} prove the following lemma relating both flow
properties:

\begin{proposition}%
  \label{lem:causal-to-focused-flow}
  If $G$ admits a causal flow, then there exists a valid focused gFlow for $G$.
\end{proposition}

\begin{lemma}%
  \label{lem:weak-to-strict-graphlike-flow}
  If $D$ is a weakly graph-like \zxGND diagram and $G(D)$ admits a causal flow,
  the strictly graph-like diagram generated by the algorithm described in
  Lemma~\ref{lem:weak-to-strict-graphlike} admits a causal flow.
\end{lemma}

\begin{proof}
  Let $(f_D, \prec_D)$ be a valid causal flow for $D$ and let $D'$ be the
  resulting diagram after applying rule~\ref{eq:toGL} over a node $v$.
  We construct a causal flow for $D'$ by defining a function $f_{D'}$ and
  relation $\prec_{D'}$ as the minimal objects such that
  \begin{itemize}
    \item $f_{D'} \supseteq f_D$ and $\prec_{D'} \supseteq \prec_D$.
    \item If $v$ is connected to an input in $D$, $f_{D'}(v_{in_1}) = v_{in_2}$,
      $f_{D'}(v_{in_2}) = v$, and $(v_{in_1},v_{in_2}),\allowbreak (v_{in_2},v) \in
      \prec_{D'}$.
    \item If $v$ is connected to at least one output in $D$ and $v$ is not a
      \ground-spider, $f_{D'}(v) = v_{out_{1,1}}$.
    \item For each output $i$ connected to $v$ in $D$,
      $(v,v_{out_{i,1}}), (v_{out_{i,1}},v_{out_{i,2}}) \in \prec_{D'}$
      and
      $f_{D'}(v_{out_{i,1}}) = v_{out_{i,2}}$.
  \end{itemize}
  Notice that $(f_{D'}, \prec_{D'}^+)$ is a valid focused flow for $D'$, where
  $\prec_{D'}^+$ is the transitive closure of $\prec_{D'}$. Since the
  rule~\ref{eq:toGL} preserves the focused flow, the resulting diagram
  after successive application admits a causal flow.
\end{proof}

\begin{lemma*}[\ref{lem:weakGL-gFlow}]%
  \textit{%
    If $C$ is a hybrid quantum-classical circuit and $D$ is the
    graph-like \zxGND-diagram obtained from the translation $\toZxGND{C}$ and
    Lemma~\ref{lem:weak-to-strict-graphlike}, then
    $G(D)$ admits a focused gFlow.
  }
\end{lemma*}

\begin{proof}
  By Proposition~\ref{lem:causal-to-focused-flow}, it suffices to prove
  that $G(D)$ admits a causal flow. We proceed by induction on the construction
  of $C$.
  \begin{itemize}
    \item Notice that the translation of each base constructors cannot be
      further simplified by rules $\SpiderRule$, \GndDoubleRule, or \ParHRule,
      and the underlying open graph trivially admits a causal flow.

    \item If $C = C_1 \otimes C_2$, let $D_i = \toZxGND{C_i}$.
      Since the two circuits are not connected after the composition,
      they do not interact via the rules $\SpiderRule$, \GndDoubleRule,
      or \ParHRule, and
      therefore $D = D_1 \otimes D_2$.
      By the induction hypothesis $G(D_1)$ and
      $G(D_2)$ admit some causal flow $(f_1, \prec_1)$ and $(f_2, \prec_2)$
      respectively. Then, $(f_1 \cup f_2, \prec_1 \cup \prec_2)$ is a causal
      flow for $G(D)$.

    \item If $C = C_1 \otimes C_2$, let $D_i = \toZxGND{C_i}$.
      By the induction hypothesis $G(D_1)$ and $G(D_2)$ admit some causal flow
      $(f_1, \prec_1)$ and $(f_2, \prec_2)$ respectively.
      Notice that rule $\SpiderRule$ will be applied between each output of
      $C_1$ and the connected inputs in $C_2$.
      Let $f'_1$ be a function and $\prec'_1$ a relation such that
      for each vertex $v$ in $G(D_1)$ and corresponding non-empty set of inputs
      $\{u_1, \dots, u_k\}$ in $G(D_2)$,
      $\forall v' st.\ f_1(v') = v, f'_1(v') = u_1$
      and $\forall v' st.\ v' \prec_1 v, \{(v',u_i)\}_{i=1}^k \subseteq \prec'_1$.
      Notice that since $D_2$ is weakly graph-like, $u_1$ has
      exactly one corresponding sink node in $G(D_1)$.
      Additionally, for all $v_1, v_2$ non-output nodes in $G(D_1)$
      let $f'_1(v_1)=v_2$ if $f_1(v_1)=v_2$
      and let $v_1\prec'_1 v_2$ if $v_1\prec_1 v_2$.
      Then, $(f'_1 \cup f_2, (\prec'_1 \cup \prec_2)^+)$
      is a causal flow for $G(D)$.
  \end{itemize}
  Then, by Lemma~\label{lem:weak-to-strict-graphlike-flow}, the application of
  Lemma~\ref{lem:weak-to-strict-graphlike} preserves the existence of a causal
  flow.
\end{proof}

\subsection{Proofs of section~\ref{sec:extraction}}%

\begin{lemma*}[\ref{lem:extraction-gflow}]%
  \textit{%
  Each step of the while loop in Algorithm~\ref{algo:extraction},
  line~\ref{algo:extraction:while}, preserves the gFlow of the diagram.
  }
\end{lemma*}

\begin{proof}
  By Proposition~\ref{lem:row-sum-extraction}, the Gauss elimination application
  preserves the gFlow.
  Then, look at the set of non-frontier spiders maximal in the gFlow order.

  If the set contains a non-\ground spider $u$ then by definition $g(u) \in
  \frontier$. If $\odd_G(g(u))$ does not contain \ground nodes, after the Gauss
  elimination there must be a frontier spider $v$ such that it is only connected
  to $u$. Therefore, removing $v$ and making $u$ a sink of the diagram does not
  break the gFlow.

  On the other case, since \ground-spiders are always sinks of the diagram
  promoting them to the frontier does not modify the gFlow of the diagram.

  Finally, the call to \textsc{CleanFrontier} does not modify the gFlow.
\end{proof}

% Another subsection for the last proofs
\input{appendix/classicalization.tex}

\section{Rule derivations}%
\label{sec:derivations}

Here we present the derivation of rules presented in
Section~\ref{sec:optimization}.

\begin{itemize}
  \item Rule~\LocalCompGndRule:
\[
  \scalebox{0.85}{\input{simpl/gnd-rewriting/gndComplement-a.tikz}}
  \ \overset{\GndDiscardRule}{=}\ 
  \scalebox{0.85}{\input{simpl/gnd-rewriting/gndComplement-1.tikz}}
  \ \overset{\LocalCompRule}{=}\ 
  \scalebox{0.85}{\input{simpl/gnd-rewriting/gndComplement-b.tikz}}
\]

%%%%%%%%%%%%%%%%%%

  \item Rule~\PivotGndRule:
\[
  \scalebox{0.85}{\input{simpl/gnd-rewriting/gndPivot-a.tikz}}
  \quad\overset{\GndDiscardRule}{=}\quad
  \scalebox{0.85}{\input{simpl/gnd-rewriting/gndPivot-1.tikz}}
\]\[
  \quad\overset{\PivotRule}{=}\quad
  \scalebox{0.85}{\input{simpl/gnd-rewriting/gndPivot-2.tikz}}
  \quad=\quad
  \scalebox{0.85}{\input{simpl/gnd-rewriting/gndPivot-b.tikz}}
\]

%%%%%%%%%%%%%%%%%%

  \item Rule~\PauliPivotGndRule:
\[
  \scalebox{0.85}{\input{simpl/gnd-rewriting/gndDeletion-a.tikz}}
  \qquad\overset{\GndDiscardRule}{=}\qquad
  \scalebox{0.85}{\input{simpl/gnd-rewriting/gndDeletion-1.tikz}}
\]\[
  \qquad\overset{\PivotRule}{=}\qquad
  \scalebox{0.85}{\input{simpl/gnd-rewriting/gndDeletion-b.tikz}}
\]

%%%%%%%%%%%%%%%%%%

  \item Rule~\RowSumGndRule:
\[
  \scalebox{0.85}{\input{simpl/gnd-row-sum-0.tikz}}
  \quad\overset{\GndCNOTRule}{=}\quad
  \scalebox{0.85}{\input{simpl/gnd-row-sum-1.tikz}}
\]
\[
  \quad\overset{\LocalCompGndRule_u}{=}\quad
  \scalebox{0.85}{\input{simpl/gnd-row-sum-2.tikz}}
  \quad\overset{\GndCNOTRule}{=}\quad
  \scalebox{0.85}{\input{simpl/gnd-row-sum-3.tikz}}
\]
\[
  \quad\overset{\LocalCompGndRule_u}{=}\quad
  \scalebox{0.85}{\input{simpl/gnd-row-sum-4.tikz}}
\]
\end{itemize}

%% file: appendix/classicalization.tex
\subsection{Proofs of Section~\ref{sec:classicalization}}%

For the following proofs we use the interpretation of \zxGND\ diagrams defined
in Appendix~\ref{sec:semantics}.

\begin{lemma}%
  \label{lem:localClassicalization:validHadam}
  The labelling rule \ClassHadamRule\ preserves the validity of the labelling.
\end{lemma}

\begin{proof}
  We prove the validity of the replacement of label $B$. The replacement of
  label $C$ is the symmetric case.

  Since the starting diagram has a valid labelling,
  $B \sqcup \hadamLbl(A)$ is a valid label if $\hadamLbl(A)$ is a valid label.
  Therefore the label is valid if
  \(\forall \rho \in \interpretLbl{A} \otimes \interpretLbl{D}\),
  \[
    \left\llbracket\input{classical/had-proof-bd.tikz}\right\rrbracket
    \circ \rho
    \in \interpretLbl{\hadamLbl(A)} \otimes \interpretLbl{D}
  \]
  Notice that this is equivalent to requiring
  $\forall a \in \interpretLbl{A}, (\H\,a\,\H^\dagger) \in
  \interpretLbl{\hadamLbl(A)}$,
  which follows from the definition of $\hadamLbl(\cdot)$.
\end{proof}

\begin{lemma}%
  \label{lem:localClassicalization:validZ}
  The labelling rule \ClassZRule\ preserves the validity of the labelling.
\end{lemma}

\begin{proof}
  By induction on the number of wires.

  \begin{itemize}
    \item If $n=0$, there are no labelled wires.
    \item If $n=1$, since the starting diagram has a valid labelling,
      $B_1 \sqcup \rot(\X)$ is a valid labelling if $\rot(\X)$ is valid.
      Notice that
      \[
        \left\llbracket\input{classical/Zn-p-proof-n1-.tikz}\right\rrbracket =
        \begin{pmatrix}
          1 & e^{-i\alpha} \\
          e^{i\alpha} & 1
        \end{pmatrix}
        \in
        \rot(\X)
      \]

    \item If $n=2$, we prove the validity of the replacement of label $B_2$. The
      replacement of label $B_1$ is the symmetric case.

      Since the starting diagram has a valid labelling,
      $B_2 \sqcup \rot(A_1)$ is a valid label if $\rot(A_1)$ is a valid label.
      Therefore the label is valid if
      \(\forall \rho \in \interpretLbl{A_1} \otimes \interpretLbl{B_1}\),
      \[
        \left\llbracket\input{classical/Z2-proof-b2.tikz}\right\rrbracket
        \circ \rho
        \in \interpretLbl{\rot(A_1)} \otimes \interpretLbl{D}
      \]
      Notice that this is equivalent to requiring
      $\forall a \in \interpretLbl{A_1}, (U\,a\,U^\dagger) \in
      \interpretLbl{\rot(A_1)}$
      for $U = \ketbra{0}{0} + e^{i\alpha}\ketbra{1}{1}$,
      which follows from the definition of $\rot(\cdot)$.

    \item If $n=3$, we prove the validity of the replacement of label $B_1$. The
      replacement of labels $B_2$ and $B_3$ are the symmetric case.

      Since the starting diagram has a valid labelling, $B_1 \sqcup \rot(A_2
      \tJoin A_3)$ is a valid label if $\rot(A_2 \tJoin A_3)$ is a valid label.
      We can split the diagram as follows, adding an intermediary label.
      \[
        \input{classical/Z3-proof-b1-alpha.tikz}
        \;\xrightarrow{\spiderrule}\;
        \input{classical/Z3-proof-b1-splitted.tikz}
      \]
      Notice that by the inductive hypothesis, the labelling step for the
      degree-2 spider is correct. Therefore it suffices to prove that the
      intermediary label $A_2 \tJoin A_3$ is valid, that is
      \(\forall \rho \in \interpretLbl{A_2 \tJoin A_3}\otimes \interpretLbl{\Q}\),
      \[
        \left\llbracket\input{classical/Z3-proof-b1.tikz}\right\rrbracket
        \circ \rho
        \in \interpretLbl{(A_2 \tJoin A_3)} \otimes \interpretLbl{\Q}
      \]
      Notice that this is equivalent to requiring
      $\forall a_2 \in \interpretLbl{A_2}, a_3 \in \interpretLbl{A_3}$,
        $(U\,(a_2 \otimes a_3)\,U^\dagger) \in \interpretLbl{A_2 \tJoin A_3}$
      for $U = \ketbra{0}{00} + \ketbra{1}{11}$,
      which follows from the definition of $\tJoin$.

    \item If $n>3$, for each $i$ and for some $k\neq i$ we can rewrite the
      diagram as follows, and apply \ClassZRule\ twice to produce the target
      $B_i \sqcup \rot(\bigstar_{j\neq i} A_j)$ label.
      \[
        \input{classical/Zn-p-proof-n4-.tikz}
        \twoheadrightarrow\;
        \input{classical/Zn-p-proof-n4b-.tikz}
      \]
      By inductive hypothesis, both rule applications produce valid labellings.
  \end{itemize}
\end{proof}

\begin{lemma*}[\ref{lem:localClassicalization}]
  \textit{The local-search labelling algorithm produces a valid labeling
  according to the standard interpretation of the \zxGND\ calculus.}
\end{lemma*}

\begin{proof}
  Notice that labelling every wire as \Q\ is always valid and hence
  the algorithm starts with a valid labelling. By
  Lemmas~\ref{lem:localClassicalization:validHadam}
  and~\ref{lem:localClassicalization:validZ}, each step applying rules
  \ClassZRule\ and \ClassHadamRule\ preserve the validity and therefore the
  final labelling is valid.
\end{proof}

%% file: figures/classical/Zn-p-proof-n1-.tikz
\begin{tikzpicture}
	\begin{pgfonlayer}{nodelayer}
		\node [style=none] (0) at (-1, 0) {};
		\node [style=Z phase dot] (1) at (-2.5, 0) {$\alpha$};
	\end{pgfonlayer}
	\begin{pgfonlayer}{edgelayer}
		\draw (0.center) to (1);
	\end{pgfonlayer}
\end{tikzpicture}

%% file: appendix/extraction.tex
\section{Extraction examples}%
\label{sec:extraction-example}

We show here an example of running the extraction procedure on a diagram. The
frontier set is demarcated with a dashed red box, and the extracted circuit is
represented with their ZX equivalent directly connected to the right of the
frontier.

We start with the frontier initialized as the output vertices, and
directly extract any \ground.
\[
  \input{extraction/step1-a.tikz}
  \;\xrightarrow{\substack{Clean\\frontier}}\;
  \input{extraction/step1-a-1.tikz}
\]

During each step of the algorithm, a maximal non-extracted element in the gFlow
order is chosen.
Candidates can be chosen efficiently without calculating the gFlow by running
Gauss elimination on the biadjacency matrix between the border and
the non-extracted spiders, and mirroring the row-sum operations using a gFlow
preserving rewrite rule on the diagram.
\[
  \;\xrightarrow{Gauss}\;
  \input{extraction/step1-b.tikz}
  \;\xrightarrow{Extract}\;
  \input{extraction/step1-c.tikz}
\]
\[
  \;\xrightarrow{Extract}\;
  \input{extraction/step1-d.tikz}
\]
Now, there are no candidate frontier spiders connected to a single unextracted
spider. We must therefore extract one of the connected \ground-spiders as a
qubit termination.
\[
  \;\xrightarrow{\substack{Extract\\\ground}}\;
  \input{extraction/step1-e.tikz}
  \;\xrightarrow{\substack{Clean\\frontier}}\;
  \input{extraction/step1-f.tikz}
\]
\[
  \;\xrightarrow{Extract}\;
  \input{extraction/step1-g.tikz}
\]

If after any step there are nodes in the frontier that are not connected to any
internal spider then they can be removed from the frontier and extracted as a
qubit initialization, as show in the following example.
\[
  \input{extraction/border-qubitInit-b.tikz}
  \;\to\;
  \input{extraction/border-qubitInit-c.tikz}
\]